\documentclass[a4paper, authoryear, 12pt]{elsarticle}

\usepackage[english]{babel}
\usepackage{babel}
\usepackage{fontenc}
\usepackage{graphicx}
\usepackage{amsmath,amsthm,amssymb}
\usepackage{natbib}
\usepackage[colorlinks=true,citecolor=blue]{hyperref}
\usepackage{color}
\usepackage{graphicx}
\usepackage{subfigure}
\usepackage{color}
\usepackage{newlfont}
\usepackage{multirow}
\usepackage{longtable} 
\usepackage{ifthen}
\usepackage{alltt}
\usepackage{enumerate}
\usepackage[text={6.25in,8.5in},centering]{geometry}
\usepackage{tikz}
\usepackage{epstopdf}
\usepackage{float}
\usepackage{arydshln}
\usepackage[utf8]{inputenc}
\usepackage{tikz}
\usetikzlibrary{shapes,arrows}
\numberwithin{equation}{section}
\newtheorem{theorem}{Theorem}[section]
\newtheorem{lemma}{Lemma}[section]
\newtheorem{proposition}{\bf Proposition}[section]
\newtheorem{remark}{Remark}[section]

\input epsf.sty
\usepackage{lineno}
\journal{arXiv}

\begin{document}
	\title{Modeling the effects of prosocial awareness on COVID-19 dynamics: A case study on Colombia}
	\author[ISI]{Indrajit Ghosh \footnote{Corresponding author. Email: indra7math@gmail.com, indrajitg\_r@isical.ac.in}}
	\author[UFL]{Maia Martcheva}
    \address[ISI]{Agricultural and Ecological Research Unit, Indian Statistical Institute, Kolkata - 700 108, West Bengal, India}
    \address[UFL]{Department of Mathematics, University of Florida, Gainesville, FL 32611, USA}	

\begin{abstract}
	The ongoing COVID-19 pandemic has affected most of the countries on Earth. It has become a pandemic outbreak with more than 24 million confirmed infections and above 840 thousand deaths worldwide. In this study, we consider a mathematical model on COVID-19 transmission with the prosocial awareness effect. The proposed model can have four equilibrium states based on different parametric conditions. The local and global stability conditions for awareness free, disease-free equilibrium is studied. Using Lyapunov function theory and LaSalle Invariance Principle, the disease-free equilibrium is shown globally asymptotically stable under some parametric constraints. The existence of unique awareness free, endemic equilibrium and unique endemic equilibrium is presented. We calibrate our proposed model parameters to fit daily cases and deaths from Colombia. Sensitivity analysis indicate that transmission rate and learning factor related to awareness of susceptibles are very crucial for reduction in disease related deaths. Finally, we assess the impact of prosocial awareness during the outbreak and compare this strategy with popular control measures. Results indicate that prosocial awareness has competitive potential to flatten the curve.
\end{abstract}


\maketitle

\section{Introduction}
The ongoing outbreak of coronavirus disease 2019 (COVID-19), caused by SARS-CoV-2 virus, a highly contagious virus, has been a massive threat for governments of many affected countries. COVID-19 is causing obstacles for public health organizations and is affecting almost every aspect of human life. The outbreak was declared a pandemic of international concern by WHO on March $11^{th}$, 2020 \cite{Who2020}. The virus can cause a range of symptoms including dry cough, fever, fatigue, breathing difficulty, and bilateral lung infiltration in severe cases, similar to those caused by SARS-CoV and MERS-CoV infections \cite{huang2020clinical, gralinski2020return}. Many people may experience non-breathing symptoms including nausea, vomiting and diarrhea \cite{cdcgov2020}. Chan et. al \cite{chan2020familial} confirmed that the virus spreads through close contact of humans. It has become an epidemic outbreak with more than $24$ million confirmed infections and above $840$ thousand deaths worldwide as of August $28^{th}$, 2020 \cite{Worldometer2020}. 

Since first discovery and identification of coronavirus in 1965, three major outbreaks occurred, caused by emerging, highly pathogenic coronaviruses, namely the 2003 outbreak of Severe Acute Respiratory Syndrome (SARS) in mainland China \cite{gumel2004modelling,li2003angiotensin}, the 2012 outbreak of Middle East Respiratory Syndrome (MERS) in Saudi Arabia \cite{de2013commentary,sardar2020realistic}, and the 2015 outbreak of MERS in South Korea \cite{cowling2015preliminary,kim2017middle}. These outbreaks resulted in SARS and MERS cases confirmed by more than $8000$ and $2200$, respectively \cite{kwok2019epidemic}. The COVID-19 is caused by a new  genetically similar corona virus to the viruses that cause SARS and MERS. Despite a relatively lower death rate compared to SARS and MERS, the COVID-19 spreads rapidly and infects more people than the SARS and MERS outbreaks. In spite of strict intervention measures implemented in various affected areas, the infection spread around the globe very rapidly. Due to nonavailability of vaccines and specific medications, non-pharmaceutical control measures such a social distancing, lockdown, use of mask, use of PPE kits, awareness through media are studied using different theoretical frameworks \cite{ngonghala2020mathematical}. 

Mathematical modeling based on differential equations may provide a comprehensive mechanism for the dynamics of the disease and also to test the efficacy of the control strategies to reduce the burden of COVID-19. Several studies were performed using real-life data from the affected countries and analyzed various features of the outbreak as well as assess the impact of intervention such as lockdown approaches to suppress the outbreak in the concerned countries \cite{kucharski2020early, sardar2020assessment, tang2020estimation,  zhao2020preliminary, asamoah2020global}. There has been a few mathematical models to assess the impact of awareness campaigns against COVID-19 \cite{yan2020impact, zhou2020effects,chang2020studying,khajanchi2020dynamics,kobe2020modeling,mbabazi2020mathematical,mohsen2020global}. These research articles mainly incorporate the awareness through media campaigns. The media effect is modelled in two ways: by adding media compartment to COVID-19 model \cite{chang2020studying, khajanchi2020dynamics,mbabazi2020mathematical,zhou2020effects, kobe2020modeling} and through reduction in incidence function due to media campaigns \cite{yan2020impact, mohsen2020global}. However, there is a scope of investigating pro-social awareness on the dynamics of COVID-19 transmission. The idea is that the aware susceptible persons will pass the information (regarding use face mask, social distancing, mortality due to COVID-19 etc.) to the unaware susceptible individuals. The unaware people become aware by contacting the aware susceptible and practice the self-protection measures. 

As a case study, we use daily notified cases and deaths in Colombia. With over $50$ million inhabitants Colombia is the third-most-populous country in Latin America. On March $6^{th}$, 2020, Colombia reported the first confirmed case of COVID-19. On $17^{th}$ March, President Iván Duque spoke to the Colombians and declared the state of emergency, announcing that he would take economic measures that were announced the following day. The first measure taken seeking the protection of the elderly is to decree mandatory isolation from $20^{th}$ March, 2020 to $31^{st}$ May, 2020 for all adults over 70 years of age. They must remain in their residences except to stock up on food or access health or financial services. Government entities were instructed to make it easier for them to receive their pensions, medicines, healthcare or food at home. On the evening of $20^{th}$ March, President Iván Duque announced a 19-day nationwide quarantine, starting on $24^{th}$ March at midnight and ending on $12^{th}$ April at midnight \cite{Wiki2020colombiacovid}. As of August $28^{th}$, 2020, there were more than $590$ thousand confirmed cases (currently, the world's $7^{th}$ highest) and above 18 thousand confirmed deaths \cite{Worldometer2020}. As the outbreak of COVID-19 is expanding rapidly in Colombia, real-time analysis of epidemiological data are required to increase situational awareness and inform interventions. Mathematical modeling based on dynamic equations \cite{pang2020transmission,tang2020updated,frank2020covid} may provide detailed mechanism for the disease dynamics. A few studies were based on the Colombia COVID-19 situation  \cite{rojas2020mathematical,bizet2020time, teheran2020epidemiological, manrique2020sir}. These studies have broadly suggested that control measures could reduce the burden of COVID-19. However, none of the studies has considered awareness as a control utilizing recent epidemic data from the Colombia.  

The main objectives of this study are to (i) propose and analyze a compartmental model incorporating prosocial awareness, (ii) use available current COVID-19 epidemic  
from Colombia and calibrate the proposed model and (iii) compare prosocial awareness with other popular control measures in Colombia. 

Rest of the paper is organized as follows: A mathematical model which incorporates the prosocial awareness is described in Section \ref{Model_formulation}. The equilibrium points of the model and their stability along with related conditions are presented in Section \ref{Mathematical_analysis}. In Section \ref{Numerical_analysis}, the transcritical bifurcation phenomenon is presented between multiple equilibria. Next in \ref{Case_study_colombia}, we fit the proposed model to daily new COVID-19 cases and deaths from Colombia. The impact of prosocial awareness and comparison with other control strategies is also studied. Finally in Section \ref{Discussion}, we discuss the results from our study.

\section{Model formulation}{\label{Model_formulation}}
A compartmental differential equation model for COVID-19 is formulated and analyzed. We adopt a variant that reflects some key epidemiological properties of COVID-19. The model monitors the dynamics of six sub-populations, namely unaware susceptible $(S_u(t))$, aware susceptible $(S_a(t))$, exposed $(E(t))$, un-notified infected $(I(t))$,  notified infected $(J(t))$ and recovered $(R(t))$ individuals. The total population size is $N(t)= S_u(t) + S_a(t) + E(t) + I(t)+J(t)+ R(t)$. Our model incorporates some demographic effects by assuming a proportional natural death rate $\mu>0$ in each of the six sub-populations of the model. In addition, our model includes a net inflow of susceptible individuals into the region at a rate $\Pi$ per unit time. This parameter includes new births, immigration and emigration. Instead of constant awareness rate, we consider that the awareness will induce a behavioral response in the person and this person will transmit the knowledge to other hosts \cite{epstein2008coupled,just2018oscillations}. Thus, unaware suscptibles can become aware through contact with aware susceptibles. The functional response in this regard is assumed to be $\frac{\alpha S_u S_a}{N}$.

 The flow diagram of the proposed model is displayed in Fig. \ref{Fig:flow_chart}.

\begin{figure}[ht]
	\centering
	\includegraphics[width=0.9\textwidth]{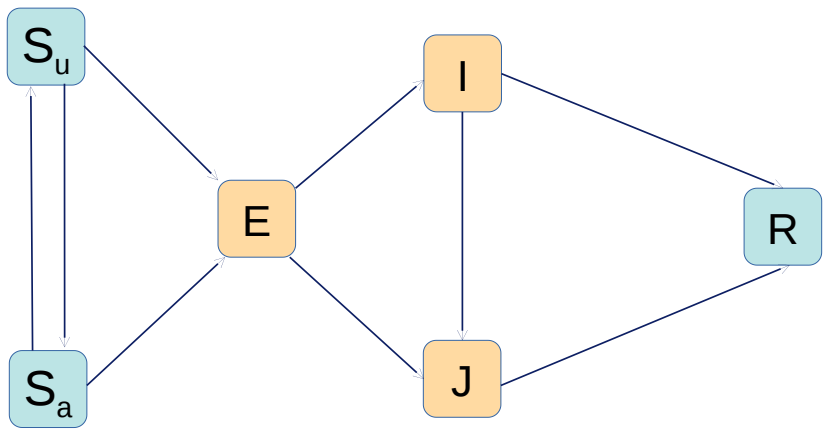}
	\caption{Compartmental flow diagram of the proposed model.}
	\label{Fig:flow_chart}
\end{figure}

Both susceptible populations decrease due to infection through successful contact with infectives who may be notified or un-notified. Note that un-notified class contains both asymptomatic and symptomatic infected individuals. We assume that the transmission co-efficient is less for notified individuals as they are kept in special observations. Thus, a reduced risk for notified COVID-19 patients is modelled as $\frac{\beta \nu J}{N}$ \cite{martcheva2015introduction}. On the other hand for un-notified individuals, the standard mixing force of infection is formulated as $\frac{\beta I}{N}$ \cite{may1991infectious}. By making successful contact with infectives both susceptible group members become exposed to the disease. The exposed population may become notified or un-notified at a rates $p \gamma$ and $(1-p) \gamma$ respectively. The recovery rate of un-notified and notified infected individuals are $\sigma_1$ and $\sigma_2$ respectively. The un-notified COVID-19 patients become notified at a rate $\eta$. The mortality rates related to COVID-19 are assumed to be $\delta_1$ and $\delta_2$ for un-notified and notified persons respectively. From the above considerations, the following system of ordinary differential equations governs the dynamics of the system:

\begin{eqnarray}\label{EQ:eqn 2.1}
\displaystyle{\frac{dS_u}{dt}} &=& \Pi- \beta \frac{I + \nu J}{N}  S_u - \alpha \frac{S_a}{N} S_u - \mu S_u +\theta S_a,\nonumber \\
\displaystyle{\frac{dS_a}{dt}} &=& \alpha \frac{S_a}{N} S_u - \epsilon \beta \frac{I + \nu J}{N} S_a - (\mu + \theta) S_a,\nonumber \\
\displaystyle{\frac{dE}{dt}} &=& \beta \frac{I + \nu J}{N} (S_u + \epsilon S_a)-(\gamma+\mu)E,\nonumber \\
\displaystyle{\frac{dI}{dt}} &=& (1-p) \gamma E - (\eta+\sigma_1+\mu + \delta_1)I, \nonumber \\
\displaystyle{\frac{dJ}{dt}} &=& p \gamma E + \eta I - (\sigma_2+\mu + \delta_2)J, \nonumber\\
\displaystyle{\frac{dR}{dt}} &=& \sigma_1 I+\sigma_2 J-\mu R, \nonumber
\end{eqnarray}
All the parameters and their biological interpretation are given in Table \ref{tab:mod1} respectively.\\
	
\begin{table}[ht]
	\begin{center}
		\caption{\textbf{Description of parameters used in the model.}} \vspace{3mm}
		\label{tab:mod1}
		\begin{tabular}{c p{7.2cm} p{2.5cm} p{2cm}}
			\hline
			\textbf{Parameters} & \textbf{Interpretation} &  \textbf{Value} & \textbf{Reference} \\ \hline
			$\Pi$ & Recruitment rate & -- & -- \\
			$\beta$ & Transmission rate & (0-1) & Estimated \\
			$\nu$ & Modification factor & 0.5 & Assumed \\
			$\alpha$ & Learning factor related to aware susceptibles & (0-1) & Estimated \\
			$\theta$ & Rate of transfer of aware individuals to unaware susceptible class & 0.02 & \cite{samanta2013effect} \\
			$\epsilon$ & Reduction in transmission co-efficient for aware susceptibles & 0.4 & Assumed\\
			$\frac{1}{\gamma}$ & Incubation period & 5 & \cite{li2020early,linton2020incubation} \\
			$p$ & Proportion of notified individuals & 0.2 & \cite{wu2020characteristics,yang2020epidemiological} \\
			$\eta$ & Transfer rate from un-notified to notified & 0.01 & Assumed\\
			$\sigma_1$ & Recovery rate from un-notified individuals & 0.17 & \cite{woelfel2020clinical,tindale2020transmission}\\
			$\sigma_2$ & Recovery rate from notified individuals & 0.072 & \cite{lopez2020end} \\
			$\delta_1$ & Disease induced mortality rate in the un-notified class & 0.01 & Assumed\\
			$\delta_2$ & Disease induced mortality rate in the notified class & (0-1) & Estimated\\
			$\mu$ & Natural death rate  & -- & -- \\
			$N$ & Total population  & -- & --\\
			\hline
		\end{tabular}
	\end{center}
\end{table}

\section{Model analysis}{\label{Mathematical_analysis}}
\subsection{Positivity and boundedness of the solution}
This subsection is provided to prove the positivity and boundedness of solutions of the system \eqref{EQ:eqn 2.1} with following initial conditions 

\begin{eqnarray}
(S_u(0),S_a(0),E(0),I(0),J(0),R(0))^T\in \mathbb{R}_{+}^6.
\label{EQ:eqn 2.2}
\end{eqnarray}

\begin{proposition}
	The system \eqref{EQ:eqn 2.1} is invariant in $\mathbb{R}_{+}^6$.
\end{proposition}
\begin{proof}
	Consider initial conditions \eqref{EQ:eqn 2.2} and let $S_u$ becomes zero at time $t_1$ before other state variables become zero, then	
	\begin{align*}
	\frac{dS_u}{dt}|_{S_u=0}&=\Pi + \theta S_a \geq 0,
	\end{align*}
	at $t_1$.
	This shows that $S_u$ is a non-decreasing function of time at $t_1$. Hence it follows that $S_u$ stays non-negative (similar argument is employed in Theorem 3.1 of Sun et al. \cite{sun2011effect}).
	For other state variables we note that
	\begin{align*}
    \frac{dS_a}{dt}|_{S_a=0}&= 0 \geq 0,\\ 
	\frac{dE}{dt}|_{E=0}&=\beta \left[\frac{(I+\nu J)}{S_u+S_a+I+J+R}\right] (S_u + \epsilon S_a)\geq 0,\\ 
	\frac{dI}{dt}|_{I=0}&= (1-p) \gamma E\geq 0,\\ 
	\frac{dJ}{dt}|_{J=0}&=p \gamma E\geq 0,\\
	\frac{dR}{dt}|_{R=0}&=\sigma_1 I + \sigma_2 J\geq 0.
	\end{align*}
	Thus we obtain non-negativity of all the six state variables and it follows that $\mathbb{R}_{+}^6$ is an invariant set for the model \eqref{EQ:eqn 2.1}.
\end{proof}

\begin{proposition}
	The system \eqref{EQ:eqn 2.1} is bounded in the region\\ $\Omega=\lbrace(S_u,S_a,E,I,J,R)\in \mathbb{R}_+^{6}|S_u+S_a+E+I+J+R\leq \frac{\Pi}{\mu}\rbrace$
\end{proposition}

\begin{proof}
	Adding all the equations of the model \eqref{EQ:eqn 2.1}, total human populations satisfy the following equations,
	\begin{align*}
	&\frac{dN}{dt}=\Pi-\mu N-\delta_1 I - \delta_2 J \leq \Pi-\mu N\\
	\end{align*}
	Since $\frac{dN}{dt}\leq \Pi-\mu N$, it follows that $\frac{dN}{dt}\leq 0$ if $N \geq \frac{\Pi}{\mu}$. 
	Thus, by using standard comparison theorem, it can be shown that $N \leq N(0) e^{-\mu t}+\frac{\Pi}{\mu}(1-e^{-\mu t})$. 
	In particular, $N (t) \leq \frac{\Pi}{\mu}$ if $N ( 0 ) \leq \frac{\Pi}{\mu}$. 
	Thus, the region $\Omega$ is positively-invariant. Further, if $N ( 0 ) > \frac{\Pi}{\mu}$, then either the solution enters $\Omega$ in finite time, or $N(t)$ approaches $\frac{\Pi}{\mu}$ asymptotically.
	Hence, the region $\Omega$ attracts all solutions in $\mathbb{R}_+^{6}$ .
\end{proof}

\subsection{Equilibrium points, threshold quantities and stability analysis}
The system \eqref{EQ:eqn 2.1} has four type of equilibrium points: awareness-free disease-free equilibrium (AFDFE), disease free equilibrium (DFE), awareness-free endemic equilibrium (AFEE) and endemic equilibrium (EE). The awareness free, DFE is given by $E_0=(\frac{\Pi}{\mu},0,0,0,0,0)$. 

\begin{lemma}
	The awareness free, DFE $E_0$ of system \eqref{EQ:eqn 2.1} is locally asymptotically stable whenever $max\left[R_1,\frac{\alpha}{\mu + \theta}\right]<1$ and unstable otherwise, where $$R_1=\frac{\beta \gamma}{(\mu + \gamma)(\sigma_2 + \mu + \delta_2)} \left[(1-p)\frac{\eta \nu + \sigma_2 + \mu + \delta_2}{\eta + \sigma_1 + \mu + \delta_1} + \nu p \right]$$.
\end{lemma}
\begin{proof}
		We calculate the Jacobian of the system \eqref{EQ:eqn 2.1} at $E_0$, which is given by
	\begin{align*}
	J_{E_{0}}={\begin{pmatrix}
		-\mu & -\alpha + \theta & 0 &  -\beta & -\beta & 0 \\
		0 & \alpha-(\mu+\theta) & 0 & 0 & 0 & 0 \\
		0 & 0 & -(\mu+\gamma) & \beta & \nu \beta & 0 \\
		0 & 0 & (1-p)\gamma & -(\eta+\sigma_1+\mu+\delta_1) & 0 & 0\\
		0 & 0 & p \gamma & \eta & -(\sigma_2+\mu+\delta_2) & 0 \\
		0 & 0 & 0 & \sigma_1 & \sigma_2 & -\mu \\
		\end{pmatrix}},
	\end{align*}
		
	Let $\lambda$ be the eigenvalue of the matrix $J_{E_{0}}$. Then the characteristic equation is given by $det(J_{E_{0}}-\lambda I)=0$.\\
	
	Clearly, $-\mu$, $-\mu$ and $\alpha - (\mu+\theta)$ are three eigenvalues of the Jacobian matrix $J_{E_0}$. The other three eigenvalues are given the following cubic equation
	\begin{eqnarray}
	\begin{array}{lll}
	f(\lambda):=\lambda^3+a_1 \lambda^{2}+a_2 \lambda+a_3=0\\
	\end{array}
	\label{EQ:eqn 2.4}
	\end{eqnarray}
	where,\\
	\begin{eqnarray}
	a_1&=& \mu + \gamma + m_1 +m_2\nonumber\\
	a_2&=& (\mu + \gamma) (m_1 +m_2) - \beta \gamma \nonumber\\
	a_3&=& m_1 m_2 (\mu+\gamma) (1-R_1).\nonumber
	\end{eqnarray}
	Here \\
	$m_1=\eta+\sigma_1+\mu+\delta_1$,\\
	$m_2=\sigma_2+\mu+\delta_2$ and \\
	$R_1=\frac{\beta \gamma}{m_2 (\mu + \gamma)} \left[(1-p)\frac{\eta \nu + m_2}{m_1} + p \nu \right]$.\\
	
It is straight forward to show that coefficients of \eqref{EQ:eqn 2.4} satisfies Routh-Hurwitz criterion if $R_1<1$. Thus, all the eigenvalues are negative or have negative real parts if in addition $\frac{\alpha}{\mu + \theta}<1$.

On the other hand, if $max\left[R_1,\frac{\alpha}{\mu + \theta}\right]>1$ then at least one eigenvalue of the Jacobian matrix is positive and $E_0$ become unstable. Hence the proof is complete.
\end{proof}

\begin{theorem}
	The awareness free DFE $E_0$ is globally asymptotically stable for the system \eqref{EQ:eqn 2.1} if $max\left[R_1,\frac{\alpha}{\mu + \theta}\right]<1$.
\end{theorem}
\begin{proof}
	The system \eqref{EQ:eqn 2.1} can be represented as
	\begin{align*}
	\frac{dX}{dt}&=F_1(X,V)\\
	\frac{dV}{dt}&=G_1(X,V), G_1(X,0)=0
	\end{align*}	
	where $X=(S_u, S_a, R)\in R_3^+$ (uninfected classes of people),
	$V=(E, I, J)\in R_3^+$ (infected classes of people), and $E{0}=(\frac{\Pi}{\mu},0,0,0,0)$ is the awareness free, DFE of the system \eqref{EQ:eqn 2.1}. The global stability of $E_0$ is guaranteed if the following two conditions are satisfied:
	
	\begin{enumerate}
		\item For $\frac{dX}{dt}=F_1(X,0)$, $X^*$ is globally asymptotically stable,
		\item $G_1(X,V) = BV-\widehat{G}_1(X,V),$ $\widehat{G}_1(X,V)\geq 0$ for $(X,V)\in \hat{\Omega}$,
	\end{enumerate}
	where $B=D_VG_1(X^*,0)$ is a Metzler matrix and $\hat{\Omega}$ is the positively invariant set with respect to the model \eqref{EQ:eqn 2.1}. Following Castillo-Chavez et al \cite{castillo2002computation}, we check for aforementioned conditions.\\
	For system \eqref{EQ:eqn 2.1},
	\begin{align*}
	F_1(X,0)&=\begin{pmatrix}
	\Pi -\mu S_u\\
	0\\
	0
	\end{pmatrix},\\
	B&=\begin{pmatrix}
	-(\gamma+\mu) & \beta & \nu \beta \\
	(1-p) \gamma & -m_1 & 0\\
	p \gamma & \eta & -m_2
	\end{pmatrix}\\
	\end{align*}
	and 
	\begin{align*}
	\widehat{G}_1(X,V)=\begin{pmatrix}
	\beta (I+\nu J) (1-\frac{S_u}{N})\\
	0\\
	0
	\end{pmatrix}.
	\end{align*}
	
	Clearly, $\widehat{G}_1(X,V)\geq 0$ whenever the state variables are inside $\Omega$. Also it is clear that $X^*=(\frac{\Pi}{\mu},0,0)$ is a globally asymptotically stable equilibrium of the system $\frac{dX}{dt}=F_1(X,0)$. Hence, the theorem follows.
\end{proof}

The unique disease-free equilibrium of the system \eqref{EQ:eqn 2.1} is given by
\begin{align*}
E_1=\left(\frac{\Pi(\mu+\theta)}{\mu \alpha},\frac{\Pi [\alpha-(\mu+\theta)]}{\mu \alpha},0,0,0,0\right),
\end{align*} 
which exists if $\frac{\alpha}{\mu+\theta}>1$.
To obtain the basic reproduction number $R_0$ of the system \eqref{EQ:eqn 2.1}, we apply the next generation matrix approach. The infected compartments of the model \eqref{EQ:eqn 2.1} consist of $( E(t), I(t), J(t))$ classes. Following the next generation matrix method, the matrix $F$ of the trransmission terms and the matrix, $V$ of the transition terms calculated at $E_1$ are,

\begin{align*}
F&=\begin{pmatrix}
0 & \beta m_3 & \nu \beta m_3 \\
0 & 0 & 0  \\
0 & 0 & 0  \\
\end{pmatrix},\\\\
V&=\begin{pmatrix}
\gamma+\mu & 0 & 0 \\
-(1-p)\gamma & m_1 & 0 \\
- p \gamma & -\eta & m_2 \\
\end{pmatrix},
\end{align*}

where,
\begin{align*}
m_1 &= \eta + \sigma_1 + \mu + \delta_1\\
m_2 &= \sigma_2 + \mu + \delta_2 \\
m_3 & = \frac{1}{\alpha} \left[ (\theta + \mu) + \epsilon \{\alpha - (\theta + \mu)\} \right]
\end{align*}

Calculating the dominant eigenvalue of the next generation matrix $FV^{-1}$ , we obtain the basic reproductive number as follows \cite{van2002reproduction}
\begin{align}\label{EQ:eqn 3.2}
R_0=\frac{m_3 \beta \gamma}{m_2 (\mu + \gamma)} \left[(1-p)\frac{\eta \nu + m_2}{m_1} +  p \nu \right]
\end{align}\\
The basic reproduction number $R_0$ is defined as the average number of secondary cases generated by one infected individual during their infectious period in a fully susceptible population. The basic reproduction number $R_0$ of \eqref{EQ:eqn 2.1} given in \ref{EQ:eqn 3.2}.

Using Theorem 2 in \cite{van2002reproduction}, the following result is established.
\begin{lemma}
	The disease-free equilibrium $\varepsilon_0$ of system \eqref{EQ:eqn 2.1} is locally asymptotically stable whenever $R_0<1$, and unstable whenever $R_0>1$.
\end{lemma}

\begin{remark}
	Note that the threshold quantities $R_1$ and $R_0$ are linearly dependent by the relation $R_0=m_3 R_1$.
\end{remark}

\begin{theorem}\label{EQ:eqn T.3.4}
	The DFE $E_1$ of the model \eqref{EQ:eqn 2.1}, is globally asymptotically stable in $\Omega$ whenever $max\left[ R_0, \frac{\beta \gamma}{(\sigma_1+\mu + \delta_1)(\mu + \gamma)},\frac{\beta \nu \gamma}{m_2(\mu + \gamma)}\right]<1$.	
\end{theorem}
\begin{proof}
	Consider the following Lyapunov function
	\begin{align*}
	\mathcal{D}=\frac{E}{\gamma + \mu}+\frac{I}{\gamma}+\frac{J}{\gamma}
	\end{align*}
	
	We take the Lyapunov derivative with respect to $t$, 
	\begin{align*}
	\dot{\mathcal{D}}&=\frac{\dot{E}}{\gamma + \mu}+\frac{\dot{I}}{\gamma}+\frac{\dot{J}}{\gamma}\\
	&=\frac{1}{\mu+\gamma}\left[ \beta \frac{I + \nu J}{N} (S_u + \epsilon S_a)\right] - \frac{(m_1 - \eta) I}{\gamma} - \frac{m_2 J}{\gamma}\\
	&\leq \frac{(\sigma_1+\mu + \delta_1) I}{\gamma} \left[ \frac{\beta \gamma}{(\sigma_1+\mu + \delta_1)(\mu + \gamma)} - 1\right] + \frac{m_2 J}{\gamma} \left[ \frac{\beta \nu \gamma}{m_2(\mu + \gamma)} - 1 \right]   \text{ (Since $S_u + \epsilon S_a \leq N$ in $\Omega$)}\\
	\end{align*}
	Thus, $\dot{\mathcal{D}}\leq 0$, whenever $max\left[\frac{\beta \gamma}{(\sigma_1+\mu + \delta_1)(\mu + \gamma)},\frac{\beta \nu \gamma}{m_2(\mu + \gamma)}\right]<1$.
	
	Since all the variables and parameters of the model \eqref{EQ:eqn 2.1} are non-negative, it follows that $\dot{\mathcal{D}}\leq 0$ with $\dot{\mathcal{D}}=0$ at DFE if $max\left[\frac{\beta \gamma}{(\sigma_1+\mu + \delta_1)(\mu + \gamma)},\frac{\beta \nu \gamma}{m_2(\mu + \gamma)}\right]<1$. Hence, $\mathcal{D}$ is a Lyapunov function on $\Omega$. Therefore, followed by LaSalle’s Invariance Principle \cite{lasalle1976stability}, that
	\begin{equation}\label{EQ:eqn 3.17}
	(E(t), I(t), J(t))\rightarrow (0,0,0) \text{ as } t \rightarrow \infty
	\end{equation}
	Since  $\lim\limits_{t\rightarrow \infty}sup I(t)=0 $ and $\lim\limits_{t\rightarrow \infty}sup J(t)=0 $ (from \ref{EQ:eqn 3.17}), it follows that, for sufficiently small $\xi_1>0, \xi_2>0$, there exist constants $L_1>0, L_2>0$ such that $\lim\limits_{t\rightarrow \infty}sup I(t)\leq \xi_1$ for all $t>L_1$ and $\lim\limits_{t\rightarrow \infty}sup J(t)\leq \xi_2$ for all $t>L_2$.\\
	Hence, it follows that,
	\begin{align*}
	\frac{dR}{dt}\leq \sigma_1 \xi_1 + \sigma_2 \xi_2 - \mu R
	\end{align*}
	Therefore using comparison theorem \cite{smith1995theory}
	\begin{align*}
	R^{\infty}=\lim\limits_{t\rightarrow \infty}sup R(t)\leq \frac{\sigma_1 \xi_1 + \sigma_2 \xi_2}{\mu}
	\end{align*}
	Therefore, as $(\xi_1,\xi_2) \rightarrow (0,0)$,  $R^{\infty}=\lim\limits_{t\rightarrow \infty}sup R(t)\leq0$\\
	Similarly by using $\lim\limits_{t\rightarrow \infty}inf I(t)=0$ and $\lim\limits_{t\rightarrow \infty}inf J(t)=0$, it can be shown that
	\begin{align*}
	R_{\infty}=\lim\limits_{t\rightarrow \infty}inf R(t)\geq 0
	\end{align*}\\
	Thus, it follows from above two relations
	\begin{align*}
	R_{\infty} \geq 0 \geq R^{\infty}
	\end{align*}
	Hence $\lim\limits_{t\rightarrow \infty} R(t)= 0$\\
	
	Substituting $E(t) = I(t) = J(t) = R(t) = 0$ in the original system \eqref{EQ:eqn 2.1}, we get

    \begin{eqnarray}\label{EQ:3.18}
    \displaystyle{\frac{dS_u}{dt}} &=& \Pi- \frac{\alpha S_u S_a}{S_u+S_a} - \mu S_u +\theta S_a,\nonumber \\
    \displaystyle{\frac{dS_a}{dt}} &=& \frac{\alpha S_u S_a}{S_u+S_a} - (\mu + \theta) S_a,\nonumber 
    \end{eqnarray}
	
	Following \cite{vargas2011global}, a suitable lyapunov function can be formulated as follows
	
	\begin{align*}
	\mathcal{L}=\left[ (S_u - S_u^*) + (S_a-S_a^*) - (S_u^* + S_a^*) ln \frac{S_u+S_a}{S_u^* + S_a^*} \right] + \frac{2 \mu (S_u^* + S_a^*)}{\alpha S_a^*} \left(S_a - S_a^* - S_a^* ln \frac{S_a}{S_a^*}\right),
	\end{align*}
	where, $S_u^* = \frac{\Pi(\mu+\theta)}{\mu \alpha}$ and $S_a^* = \frac{\Pi [\alpha-(\mu+\theta)]}{\mu \alpha}$.\\
	
	Therefore by combining all above equations, it follows that each solution of the model equations \eqref{EQ:eqn 2.1}, with initial conditions $\in \Omega$ , approaches $E_1$ as $t\rightarrow \infty $ for  $max\left[\frac{\beta \gamma}{(\sigma_1+\mu + \delta_1)(\mu + \gamma)},\frac{\beta \nu \gamma}{m_2(\mu + \gamma)}\right]<1$. 
\end{proof}

\subsubsection{Existence of awareness-free endemic equilibrium} 
Let $E_2=(S_u^*, 0, E^*, I^*, J^*, R^*)$ be any AFEE of system \eqref{EQ:eqn 2.1}. Let us denote
\begin{align*}
m_1&=\eta+\sigma_1+\mu+\delta_1,\\
m_2&=\sigma_2+\mu+\delta_2,\\
m_4&=\frac{(1-p)\gamma}{m_1},\\
m_5&=\frac{p \gamma}{m_2} + \frac{\eta (1-p)\gamma}{m_1 m_2}.
\end{align*}

Further, the force of infection be 
\begin{align}\label{EQ:eqn 3.33}
\lambda_h^*=\frac{\beta [I^* + \nu J^*]}{N^*}
\end{align}
By setting the right equations of system \eqref{EQ:eqn 2.1} equal to zero, we have
\begin{align}\label{EQ:eqn 3.43}\nonumber
S_u^*&=\frac{\Pi}{\lambda_h^* + \mu},\\\nonumber 
E^*&=\frac{\lambda_h^*S_u^*}{\gamma + \mu},
I^*=m_4 E^*, \\\nonumber
J^*&=m_5 E^*,
R^* =\frac{\sigma_1 m_4 E^* + \sigma_2 m_5 E^*}{\mu},\\
N^* &=\frac{\Pi - \delta_1 m_4 E^* - \delta_2 m_5 E^*}{\mu}.
\end{align}

After simplification, we have the expression of $E^*$ as follows:
$$ E^* = \frac{\Pi}{(\beta - \delta_1) m_4 + (\beta \nu - \delta_2)m_5}\left[ R_1 - 1 \right] $$, where $R_1$ is same as the threshold quantity for the AFDFE, given by

$$R_1=\frac{\beta \gamma}{m_2 (\mu + \gamma)} \left[(1-p)\frac{\eta \nu + m_2}{m_1} + p \nu \right]$$.

Therefore, the AFEE will exist if $R_1>1$ and $\beta > max\{ \delta_1, \frac{\delta_2}{\nu}\}$.

\subsubsection{Existence of endemic equilibrium} 
Let $E_{**}=(S_u^{**}, S_a^{**}, E^{**}, I^{**}, J^{**}, R^{**})$ be any endemic equilibrium of system \eqref{EQ:eqn 2.1}. Let us denote
\begin{align*}
m_1&=\eta+\sigma_1+\mu+\delta_1, m_2=\sigma_2+\mu+\delta_2,\\
m_3 &= \frac{1}{\alpha} \left[ (\theta + \mu) + \epsilon \{\alpha - (\theta + \mu)\} \right], m_4=\frac{(1-p)\gamma}{m_1},\\
m_5&=\frac{p \gamma}{m_2} + \frac{\eta (1-p)\gamma}{m_1 m_2}, m_6= \beta (m_4 + \nu m_5),
m_7= \frac{\mu + \theta}{\alpha \mu}.
\end{align*}

Further, the force of infection be 
\begin{align}\label{EQ:eqn 3.3}
\lambda_h^{**}=\frac{\beta [I^{**} + \nu J^{**}]}{N^{**}}
\end{align}
By setting the right equations of system \eqref{EQ:eqn 2.1} equal to zero, we have
\begin{align}\label{EQ:eqn 3.4}\nonumber
S_u^{**}&=\frac{N^{**}}{\alpha} \left( \epsilon \lambda_h^{**} + \mu + \theta \right),\nonumber 
S_a^{**}=\frac{\Pi \alpha - N^{**}( \lambda_h^{**} + \mu)( \epsilon \lambda_h^{**} + \mu + \theta)}{\alpha( \epsilon \lambda_h^{**} + \mu)},\\\nonumber 
E^{**}&=\frac{\lambda_h^{**}(S_u^{**} + \epsilon S_a^{**})}{\gamma + \mu},
I^{**}=m_4 E^{**}, \\\nonumber
J^{**}&=m_5 E^{**},
R^{**} =\frac{\sigma_1 m_4 E^{**} + \sigma_2 m_5 E^{**}}{\mu},\\
N^{**} &=\frac{\Pi - \delta_1 m_4 E^{**} - \delta_2 m_5 E^{**}}{\mu}.
\end{align}

From Equations \eqref{EQ:eqn 3.3} and \eqref{EQ:eqn 3.4}, we have
\begin{align}\label{EQ:eqn 3.5}
\lambda_h^{**}&=\frac{\beta (m_4 + \nu m_5) E^{**}}{N^{**}}\\ \nonumber
&=\frac{\beta (m_4 + \nu m_5) \lambda_h^{**}}{\mu + \gamma} \left(\frac{S_u^{**}}{N^{**}} + \epsilon \frac{S_a^{**}}{N^{**}}\right)
\end{align}

This implies

\begin{align}\label{EQ:eqn 3.6}
\frac{\mu + \gamma}{\beta (m_4 + \nu m_5)}=\frac{1}{\alpha} \left( \epsilon \lambda_h^{**} + \mu + \theta \right) + \frac{\Pi \alpha \epsilon }{\alpha( \epsilon \lambda_h^{**} + \mu) N^{**}} - \frac{\epsilon( \lambda_h^{**} + \mu)( \epsilon \lambda_h^{**} + \mu + \theta)}{\alpha( \epsilon \lambda_h^{**} + \mu)}
\end{align}

Now, using expression of $N^{**}$ and equation \eqref{EQ:eqn 3.5}, we have

\begin{align}\label{EQ:eqn 3.7}
\frac{\Pi}{N^{**}}=\mu + \frac{(\delta_1 m_4 + \delta_2 m_5)\lambda_h^{**}}{\beta (m_4 + \nu m_5)} 
\end{align}

Putting the value of $\frac{\Pi}{N^{**}}$ in equation \eqref{EQ:eqn 3.6} and simplifying, we obtain

\begin{align}\label{EQ:eqn 3.8}
\lambda_h^{**}=\frac{\alpha \mu (\mu + \gamma)(R_0 - 1)}{\epsilon (\mu + \gamma)\{\alpha - \mu (1 - \epsilon) R_1\} - \alpha \epsilon (\delta_1 m_4 + \delta_2 m_5)} 
\end{align}

This indicate that the model \ref{EQ:eqn 2.1} has unique endemic equilibrium if it exists. The existence criterion for the EE is $R_0>1$ and $$\alpha(\mu + \gamma)>\mu (\mu + \gamma)(1 - \epsilon) R_1 + \alpha (\delta_1 m_4 + \delta_2 m_5)$$.

The stability analysis of the equilibria for the model (\ref{EQ:eqn 2.1}) is summarized in Table \ref{Tab:stability_existence}.

\begin{table}[ht]
	\centering
	\caption{\textbf{The local stability of equilibria for the Model (\ref{EQ:eqn 2.1}). LAS $\equiv$ Locally Asymptotically Stable and GAS $\equiv$ Globally Asymptotically Stable.}} \vspace{3mm}
	\begin{tabular}{p{2cm}p{5cm}p{8cm}}
		\hline
		\textbf{Equilibria} &  \textbf{Existence condition} & \textbf{Stability criterion}      \\
		\hline
		$E_0$ & Always exits & LAS as well as GAS if $max\left[R_1,\frac{\alpha}{\mu + \theta}\right]<1$  \\
		\hline
		$E_1$ & $\alpha>\theta+\mu$ & LAS if $R_0<1$ and GAS if $max\left[ R_0, \frac{\beta \gamma}{(\sigma_1+\mu + \delta_1)(\mu + \gamma)},\frac{\beta \nu \gamma}{m_2(\mu + \gamma)}\right]<1$ \\
		\hline
		$E_2$ & $\beta > max\{ \delta_1, \frac{\delta_2}{\nu}\}$ and $R_1>1$ &    \\
		\hline
		$E_3$ & $R_0>1$ and $\alpha(\mu + \gamma)>\mu (\mu + \gamma)(1 - \epsilon) R_1 + \alpha (\delta_1 m_4 + \delta_2 m_5)$ &  \\
		\hline
	\end{tabular}
	\label{Tab:stability_existence}
\end{table}

\section{Numerical bifurcation analysis}\label{Numerical_analysis}
In this section, various possibilities of forward transcritical bifurcations are examined based on the stability analysis of the four equilibrium points of the model \eqref{EQ:eqn 2.1}. To do the numerical experiments, the initial conditions are assumed to be,
\begin{align*}
N &= 10000, S_u(0)=0.9 \times N,\\
S_a(0) &= 100, E(0)=100, \\
I(0) &= 10, J(0) = 10, R(0)=0. 
\end{align*}

The fixed parameters used in this section are as follows: $\Pi=\mu \times N$, $\nu=0.05$, $\gamma=0.1$, $\eta=0.01$, $\sigma_1=0.05$, $\sigma_2=0.01$, $p=0.2$, $\delta_1=0.01$ and $\delta_2=0.003$.

From the stability analysis of the equilibia $E_0$ and $E_1$, it can be inferred that the $E_0$ will become unstable after a certain threshold of $\alpha$ namely $\mu+\theta$ whenever $R_1<1$. In this region when $\alpha>\mu+\theta$, the DFE will become LAS whenever $R_0<1$. This phenomenon is depicted in Fig. \ref{Fig:transcritical}(a). This type of phenomenon is called forward transcritical bifurcation where the two equilibrium points switch their stability at a critical value. 

\begin{figure}[ht]
	\includegraphics[width=0.45\textwidth]{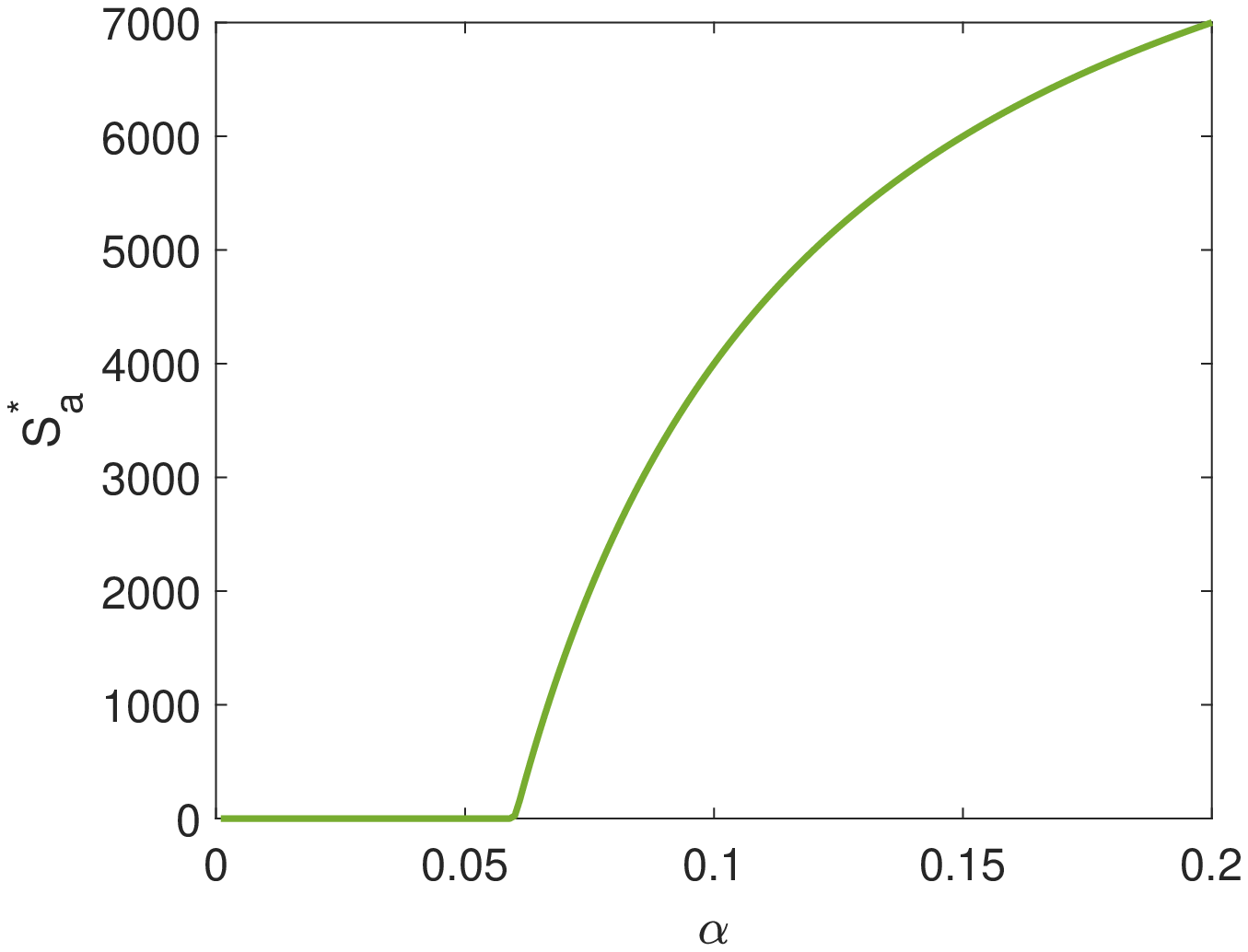}(a)
	\includegraphics[width=0.45\textwidth]{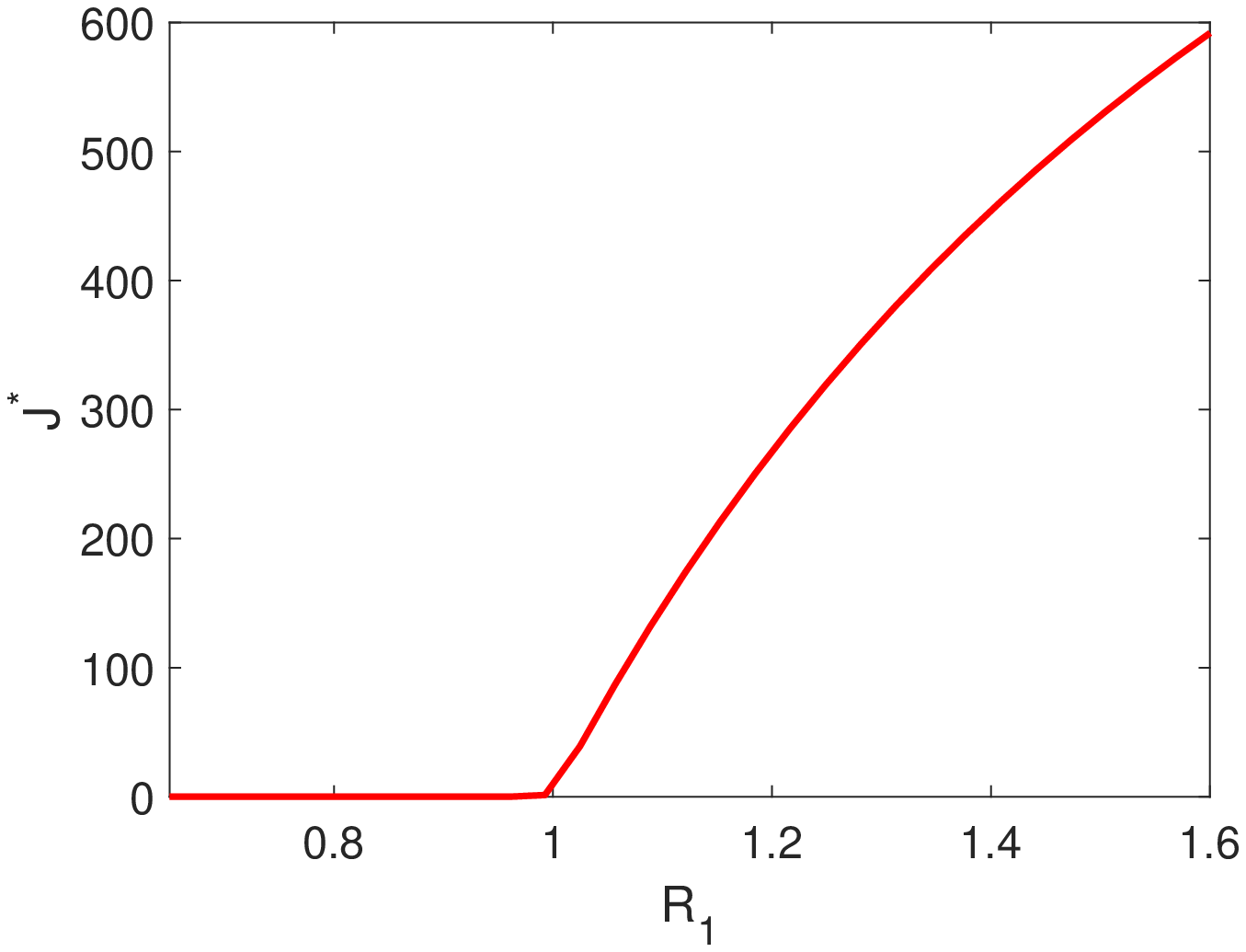}(b)
	\centering
	\includegraphics[width=0.45\textwidth]{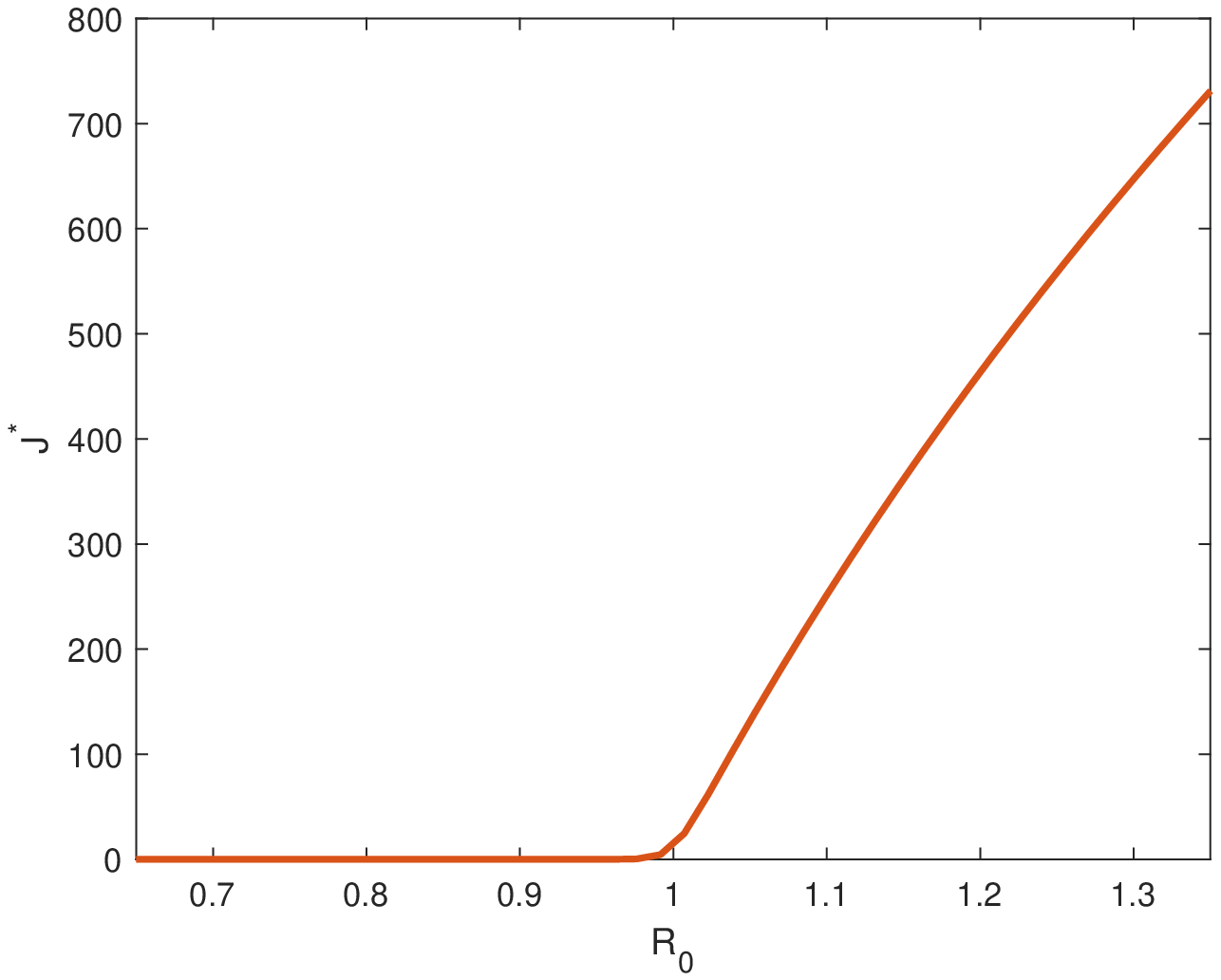}(c)
	\caption{(a) Transcritical bifurcation for the aware susceptibles at equilibrium ($S_a^*$) of the model \eqref{EQ:eqn 2.1}. Using the parameter values: $\mu=0.01$, $\beta=0.05$, $\theta=0.05$, $\epsilon=0.04$ and $0.001<\alpha<0.2$. For this parameter set, $R_1=0.4842$. (b) Transcritical bifurcation for the notified infected population at equilibrium ($J^*$) of the model \eqref{EQ:eqn 2.1}. Using the parameter values: $\mu=0.03$, $\alpha=0.03$, $\theta=0.03$, $\epsilon=0.04$ and $0.1<\beta<0.25$. (c) Transcritical bifurcation for the notified infected population at equilibrium ($J^*$) of the model \eqref{EQ:eqn 2.1}. Using the parameter values: $\mu=0.01$,  $\alpha=0.3$, $\theta=0.01$, $\epsilon=0.1$ and $0.4<\beta<0.9$.}
	\label{Fig:transcritical}
\end{figure}

Further, the forward transcritical bifurcation between the equilibria $E_0$ and $E_2$ is depicted in Fig. \ref{Fig:transcritical}(b). The reason behind this result is that the existence of $E_2$ depend on the threshold quantity $R_1$ of $E_0$. This result indicate that depending on parameter values, the awareness free, DFE can become unstable and one of $E_1$ or $E_2$ will become stable. 

Furthermore, the forward transcritical bifurcation between the equilibria $E_1$ and $E_3$ is depicted in Fig. \ref{Fig:transcritical}(c). The reason behind this result is that the existence of $E_3$ depend on the threshold quantity $R_0$ of $E_1$. This result indicate that depending on parameter values, the DFE can become unstable and $E_3$ will become stable.

\section{Case study on Colombia COVID-19 data}\label{Case_study_colombia}
\subsection{Data description}
As of August, $28^{th}$ 2020, there were more than $590$ thousand cases and above $18$ thousand deaths in Colombia. Daily COVID-19 notified cases and deaths of Colombia for the time period March $19^{th}$, 2020 to August $24^{th}$, 2020 is considered for our study. These 159 days COVID-19 notified cases and deaths were collected from \cite{Worldometer2020}. We use the first 149 data points to calibrate the unknown model parameters. In this time period, COVID-19 cases and deaths both display a upward trend for Colombia. This is an alarming situation as the pandemic continue to affect the country. This is why we chose Colombia for the case study. However, we use the remaining 10 data points to check the accuracy of the fitted model. The demographic parameter values and initial conditions for fitting the proposed model to data are given in Table \ref{Tab:demo_init-param-Table}.

\begin{table}[ht]	
	\centering
	\caption{\bf{Demographic parameter values for Colombia and initial conditions}}. \vspace{0.3cm}
	\begin{tabular}{cp{5cm}cc} \hline
		\textbf{Parameters/IC's} & \textbf{Description} & \textbf{Values} & \textbf{Reference} \\ \hline
		$N$ & Total population of Colombia & 50951997 & \cite{Worldometer2020}\\
		$\mu$ & Natural death rate or (life expectancy)$^{-1}$ & 0.3518 $\times$ $10^{-4}$ & \cite{Worldometer2020}\\  
		$\Pi$ & Recruitment rate &  & $\mu \times N$\\
		$S_u(0)$ & Initial number of unaware susceptible  & $0.9 \times N$ & -- \\ 
		$S_a(0)$ & Initial number of aware susceptible & 100 & -- \\
		$J(0)$ & Initial number of notified patients & 2 & Data\\ 
		$R(0)$ & Initial number of recovered patients & 0 & -- \\
		\hline
	\end{tabular}
	\label{Tab:demo_init-param-Table}
\end{table}

\subsection{Model calibration}
We fit the model \eqref{EQ:eqn 2.1} to daily new notified cases of COVID-19 for Colombia. Fixed parameters of the model \eqref{EQ:eqn 2.1} are given in Table \ref{tab:mod1}. The demographic parameters related to Colombia and initial condition are reported in Table \ref{Tab:demo_init-param-Table}. We estimate three unknown model parameters such as: (a) the transmission rate of infected individuals ($\beta$), (b) ($\alpha$) and (c) ($\delta_2$) by fitting the model to newly daily reported cases. Additionally, two initial conditions of the model \eqref{EQ:eqn 2.1} were also estimated from the data, namely initial number of exposed individuals $E(0)$ and initial number of un-notified individuals $I(0)$. During the specified time period, nonlinear least square solver $lsqnonlin$ (in MATLAB) is used to fit simulated daily data to the reported COVID-19 cases and deaths in Colombia. We used Delayed Rejection Adaptive Metropolis algorithm \cite{haario2006dram} to generate the 95\% confidence region. An explanation of this technique for model fitting is given in \cite{ghosh2017mathematical}. The estimated parameters are given in Table \ref{Tab:estimated-parameters-Table}. The fitting of the daily new notified COVID-19 cases and deaths of Colombia is displayed in Fig. \ref{Fig:fitting_new_cases}. 

\begin{figure}[ht]
	\centering
	\includegraphics[width=0.49\textwidth]{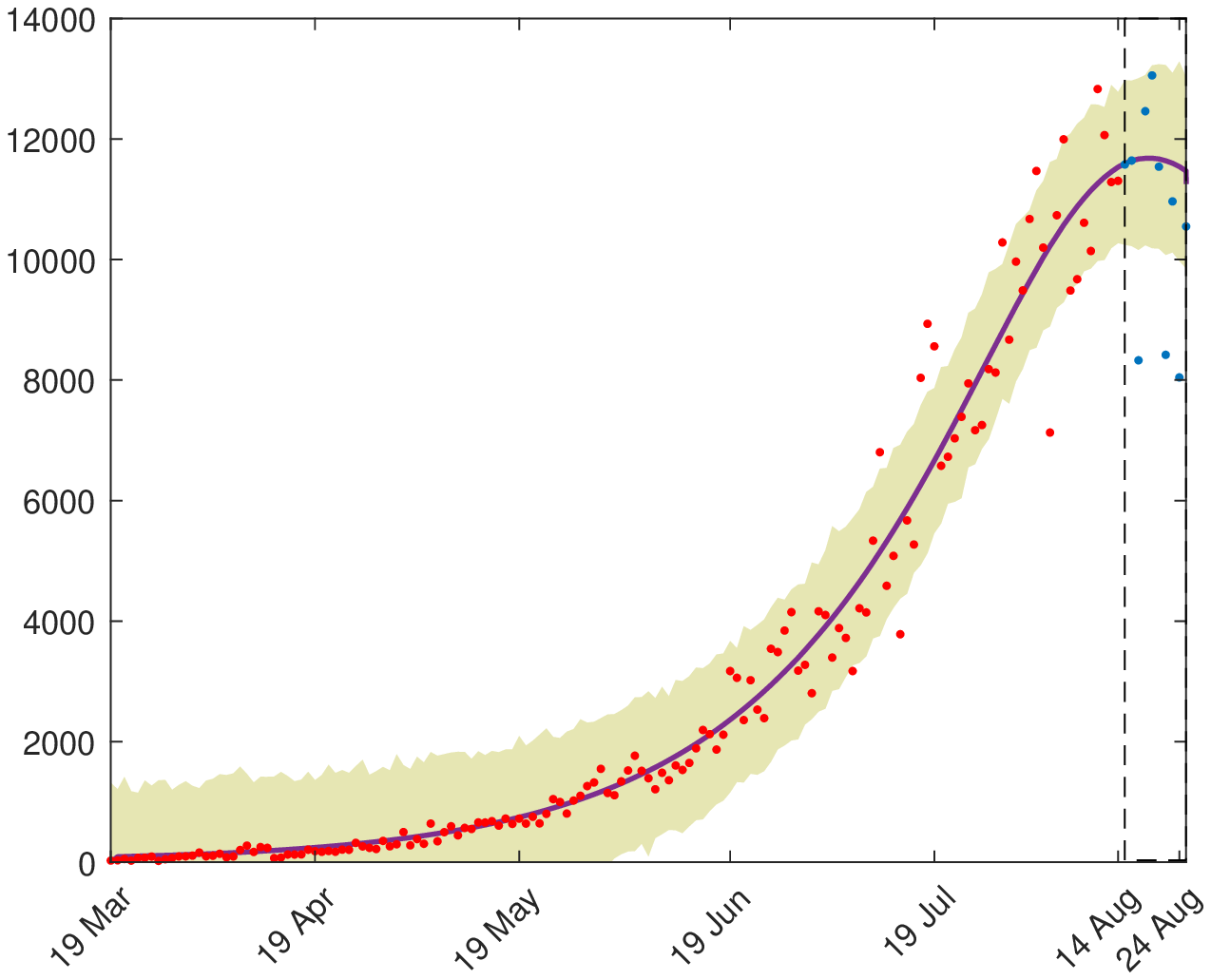}
	\includegraphics[width=0.49\textwidth]{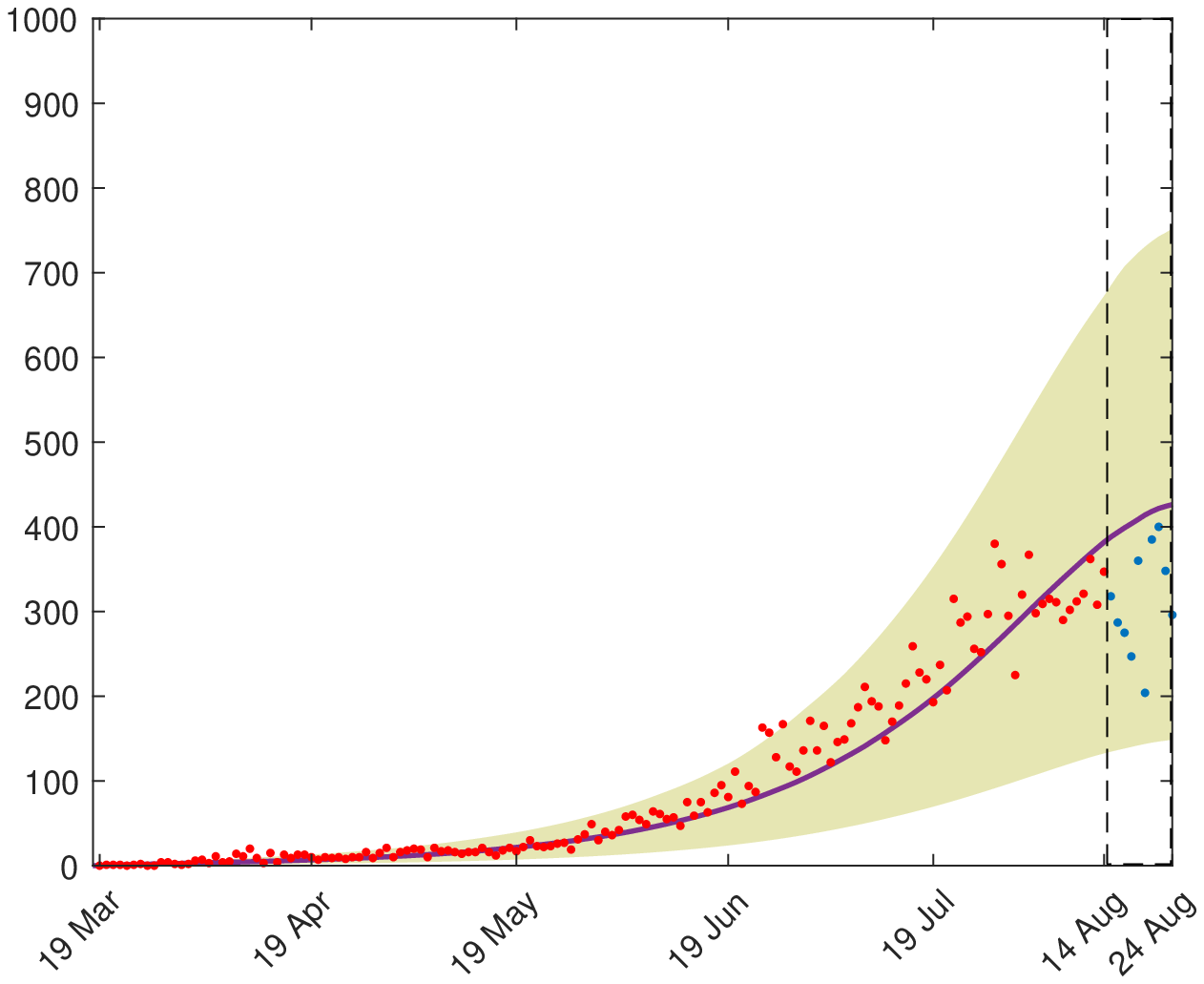}
	\caption{Model~(\ref{EQ:eqn 2.1}) fitting to daily notified COVID-19 cases and notified deaths due to COVID-19 in Colombia. Daily notified cases (deaths) are depicted in red dots and purple curve is the model simulation. The blue dots are test data points in both panels. Grey shaded region is the 95\% confidence region.}
	\label{Fig:fitting_new_cases}
\end{figure}

\begin{table}[ht]	
	\centering
	\caption{\bf{Estimated parameter values of the model~(\ref{EQ:eqn 2.1})}}. \vspace{0.3cm}
	\begin{tabular}{p{3.5cm}cc} \hline
		\textbf{Parameters} & \textbf{Mean values} & \textbf{ 95\% confidence interval} \\ \hline
		$\beta$ & 0.2584 & (0.2514 - 0.2662)\\ 
		$\alpha$ & 0.1069 & (0.1036 - 0.1099)\\  
		$\delta_2$ & 0.0032 & (0.0011 - 0.0054)\\
		$E(0)$ & 1919 & (1175 - 4147)\\ 
		$I(0)$ & 1668 & (984 - 3136)\\
		\hline
	\end{tabular}
	\label{Tab:estimated-parameters-Table}
\end{table}

Using test data points (August $15^{th}$, 2020 to August $24^{th}$, 2020), we calculate two accuracy metrics for these 10 test data points: Root mean squared error (RMSE) and mean absolute error (MAE). For the case fitting, we found that RMSE=2031.5 and MAE=1491.1. On the other hand for fitting COVID-19 deaths, RMSE=137.08 and MAE=125.05 are found. This indicate that the fitting are quite good for both scenarios.

Finally, we estimate the basic reproduction number $(R_0)$, for the proposed model \eqref{EQ:eqn 2.1}. We draw 1000 samples of the estimated parameters from their posterior distribution obtained from the MCMC run and put them in the expression of $R_0$. All the fixed parameters are taken from the Table \ref{tab:mod1} and Table \ref{Tab:demo_init-param-Table}. Estimated values of $R_0$ is found to be 0.7815 with 95\% confidence interval (0.7633 -- 0.8014).

\subsection{Sensitivity analysis}
We performed global sensitivity analysis to identify most influential parameters with respect to the total deaths due to COVID-19 in 6 months time frame (starting from March, $19^{th}$ 2020). Let us denote by $D_{total}$ the total number notified and un-notified deaths due to COVID-19. Partial rank correlation coefficients (PRCCs) are calculated and plotted in Fig. \ref{Fig:PRCCs}. Nonlinear and monotone relationship were observed for the parameters with respect to $D_{total}$, which is a prerequisite for performing PRCC analysis. Following Marino et. al \cite{marino2008methodology}, we calculate PRCCs for the parameters $\beta$, $\alpha$, $\theta$, $\nu$, $\epsilon$, $\delta_1$, $\delta_2$ and $\eta$. The following response function is used to calculate the PRCC values 
$$D_{total}=\int_0^T [\delta_1 I(t) + \delta_2 J(t)]\textup{d}t,$$
where T=180 days (chosen arbitrarily). The base values for the parameters $\beta$, $\alpha$ and $\delta_2$ are taken as the mean of estimated parameters reported in Table \ref{Tab:estimated-parameters-Table}. The other base values are taken as the fixed values displayed in Table \ref{tab:mod1}. For each of the parameters, 500 Latin Hypercube Samples were generated from the interval (0.5 $\times$ base value, 1.5 $\times$ base value).

\begin{figure}[t]
	\centering
	\includegraphics[width=0.7\textwidth]{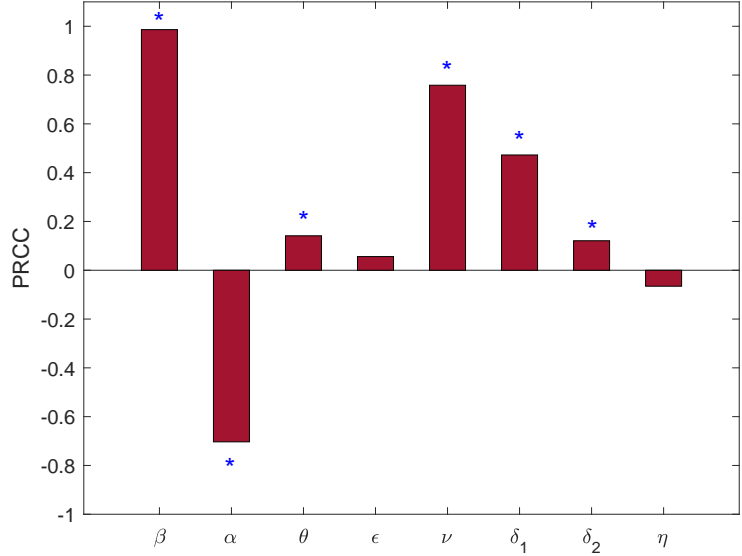}
	\caption{Effect of uncertainty of the model (\ref{EQ:eqn 2.1}) on the total number of deaths due to COVID-19 in Colombia. Parameters with significant PRCC indicated as $^*$ (p-value $<$ 0.05). The fixed parameters are taken from Table \ref{tab:mod1}.}
	\label{Fig:PRCCs}
\end{figure}

It is observed that the parameters $\beta$, $\nu$, $\delta_1$, $\theta$ and $\delta_2$ have significant positive correlations with $D_{total}$. This indicates that transmission rate of COVID-19 will increase the total number deaths related to COVID-19. Besides, modification factor for notified patients and death rate of un-notified patients are positively correlated with $D_{total}$. On the other hand, the learning factor related to aware susceptibles has significant negative correlation with the response variable. These results reinforces the fact that $\beta$, $\nu$ and $\alpha$ are very crucial for reduction of COVID-19 cases in Colombia. 

\subsection{Future projections and control scenarios}
In this section, we focus on four controllable parameters namely learning factor related to aware susceptibles ($\alpha$), transmission rate ($\beta$), modification factor for notified patients or equivalently the efficacy of notified case containment ($\nu$) and reduction in transmission co-efficient for aware susceptibles ($\epsilon$). Transmission rate of COVID-19 can be reduced by social distancing, face mask use, PPE kit use and through use of alcohol based hand wash \cite{li2020early,cdcgov2020}. The parameters $\nu$ and $\epsilon$ can also be reduced through effective management of notified cases and through increased behavioral changes by aware susceptible respectively. On the other hand, $\alpha$ should be increased to reduce the burden of COVID-19 in the society. It can be increased if the aware people show more pro-social activities. However, the results from global sensitivity analysis suggest that $\beta$ is most effective in terms of reduction in total COVID-19 related deaths. The parameters $\alpha$ and $\nu$ were also found significant with respect to the response function. Now we visualize the impacts of these four parameters on the un-notified and notified COVID-19 cases in Colombia. Using the estimated parameters (see Table \ref{Tab:estimated-parameters-Table}) we predict un-notified and notified infections
in the coming 5 months (150 days) starting from August $15^{th}$, 2020. The baseline curve is determined by simulating the model with fixed parameters from Table \ref{tab:mod1} and mean values of estimated parameters from Table \ref{Tab:estimated-parameters-Table}. For different values of $\alpha$, $\beta$, $\nu$ and $\epsilon$, the case reduction in COVID-19 cases is
depicted in Fig. \ref{Fig:prevalance1} and \ref{Fig:prevalance2}.

\begin{figure}[ht]
	\centering
	\includegraphics[width=0.49\textwidth]{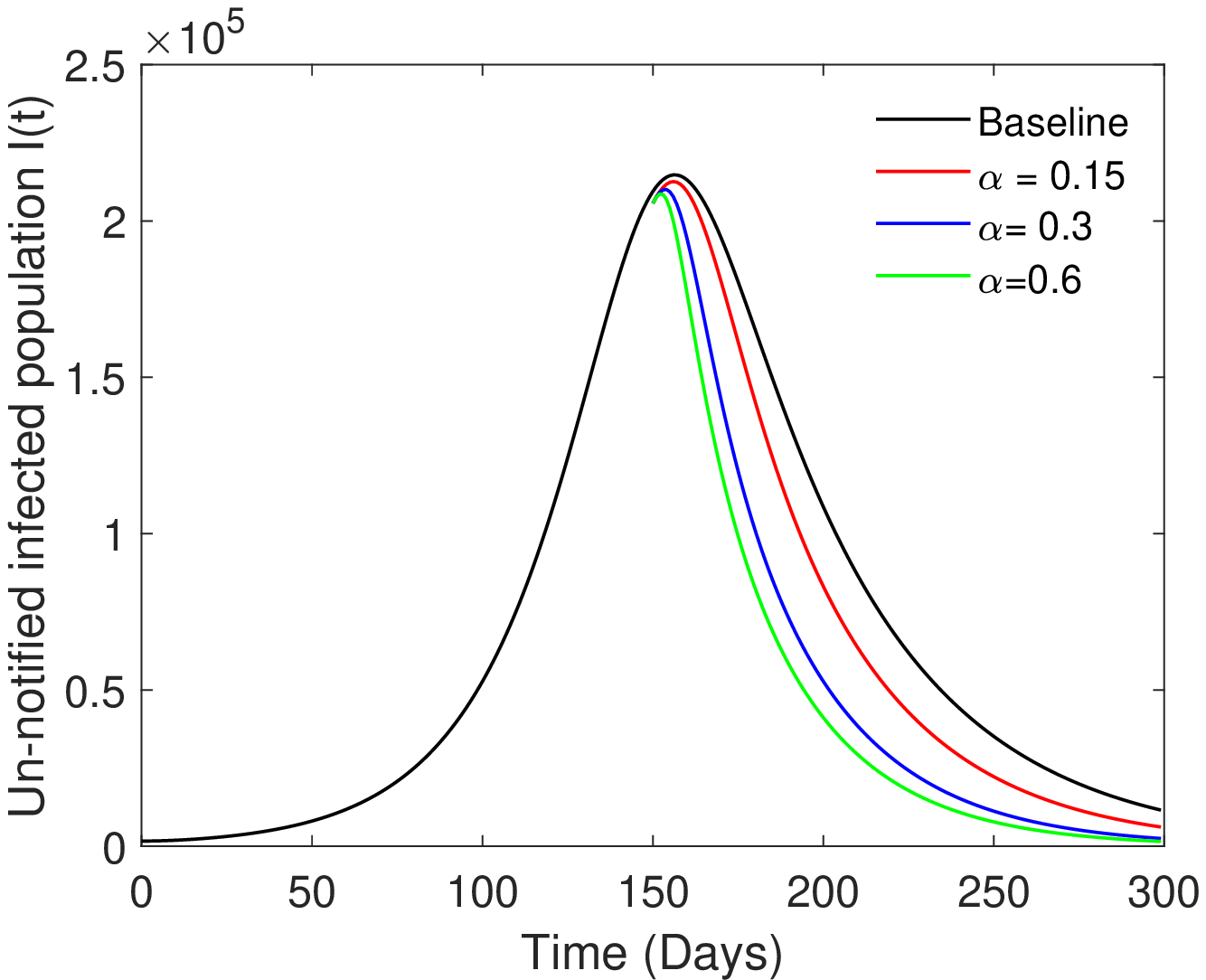}
	\includegraphics[width=0.49\textwidth]{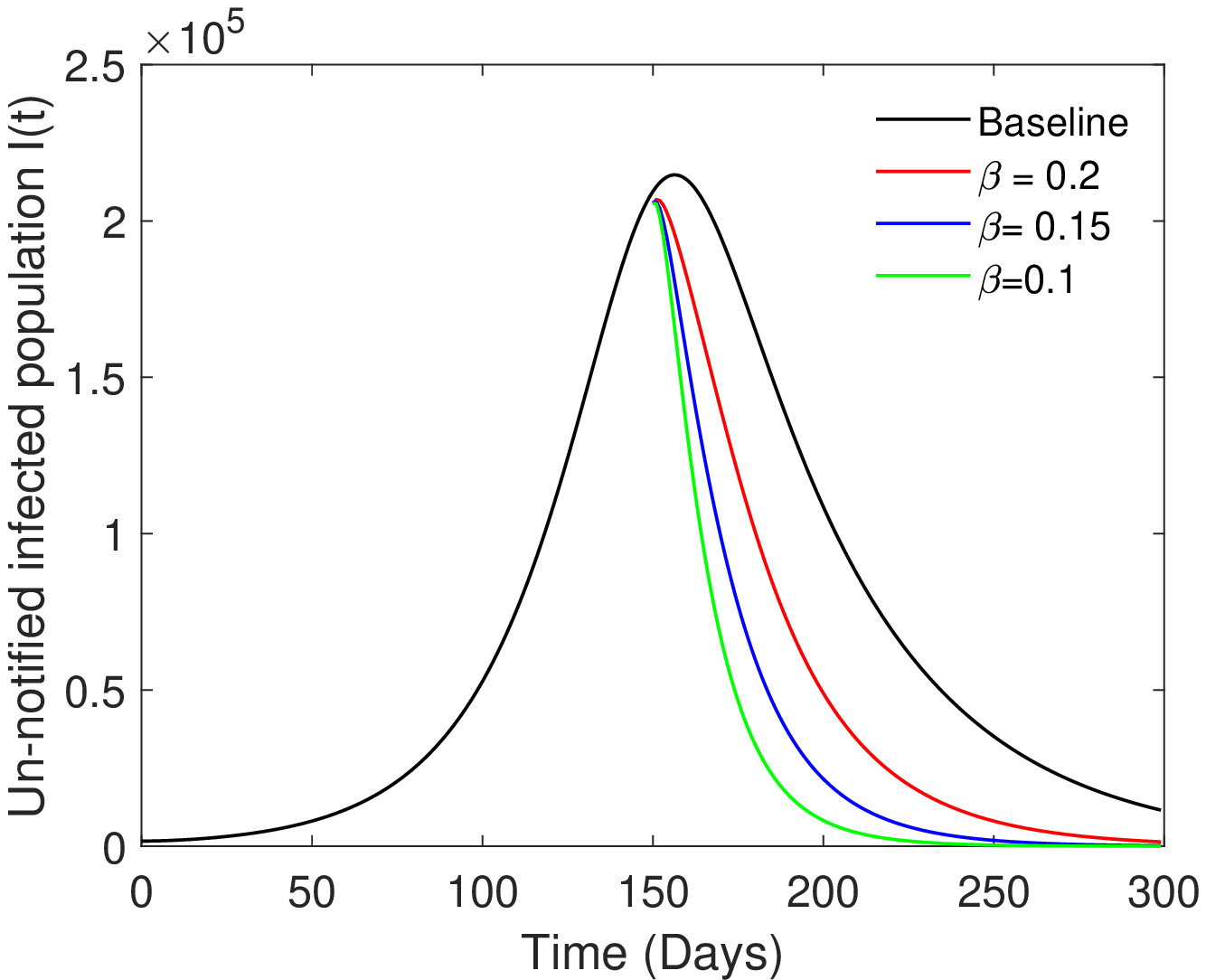}
	\includegraphics[width=0.49\textwidth]{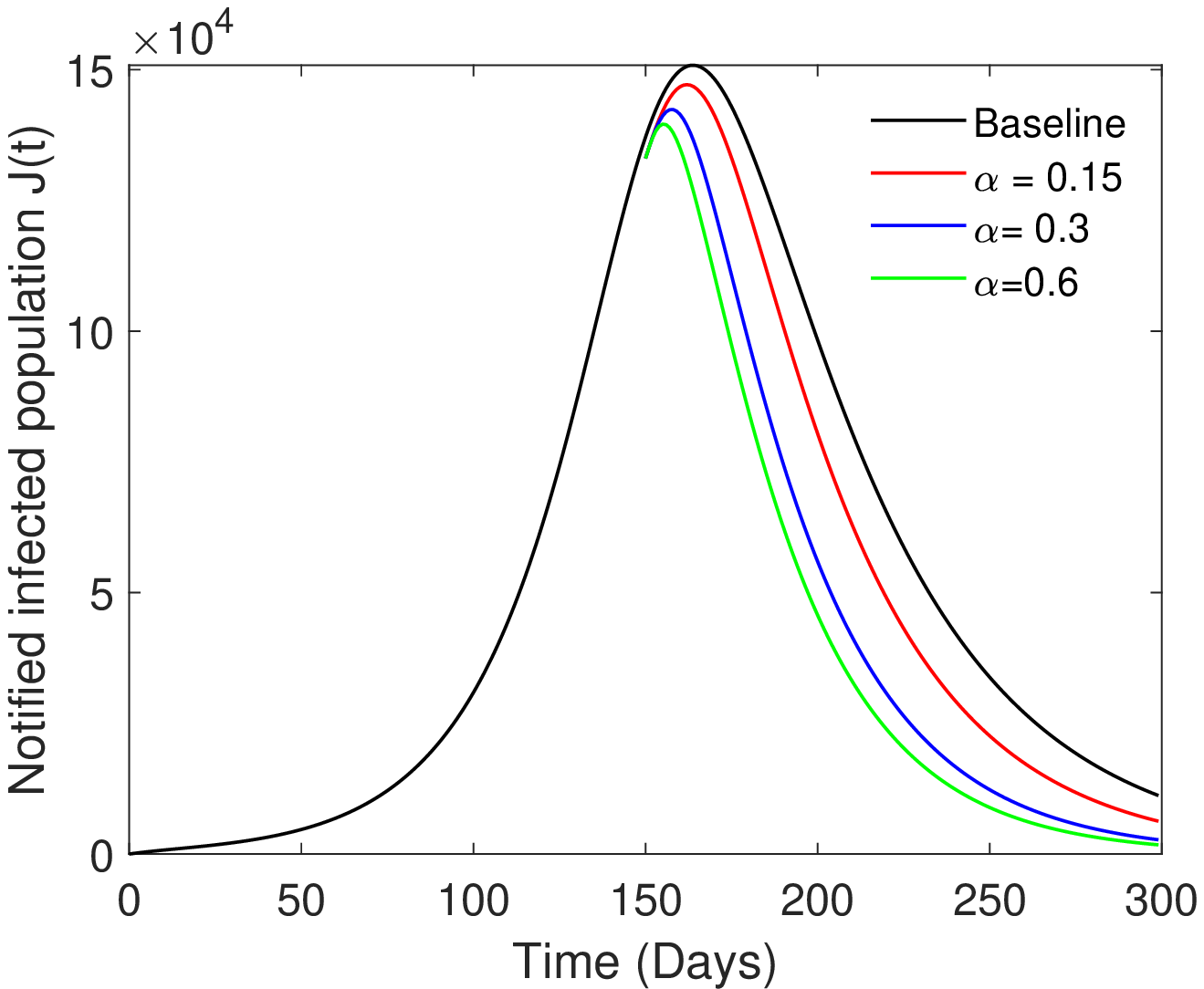}
	\includegraphics[width=0.49\textwidth]{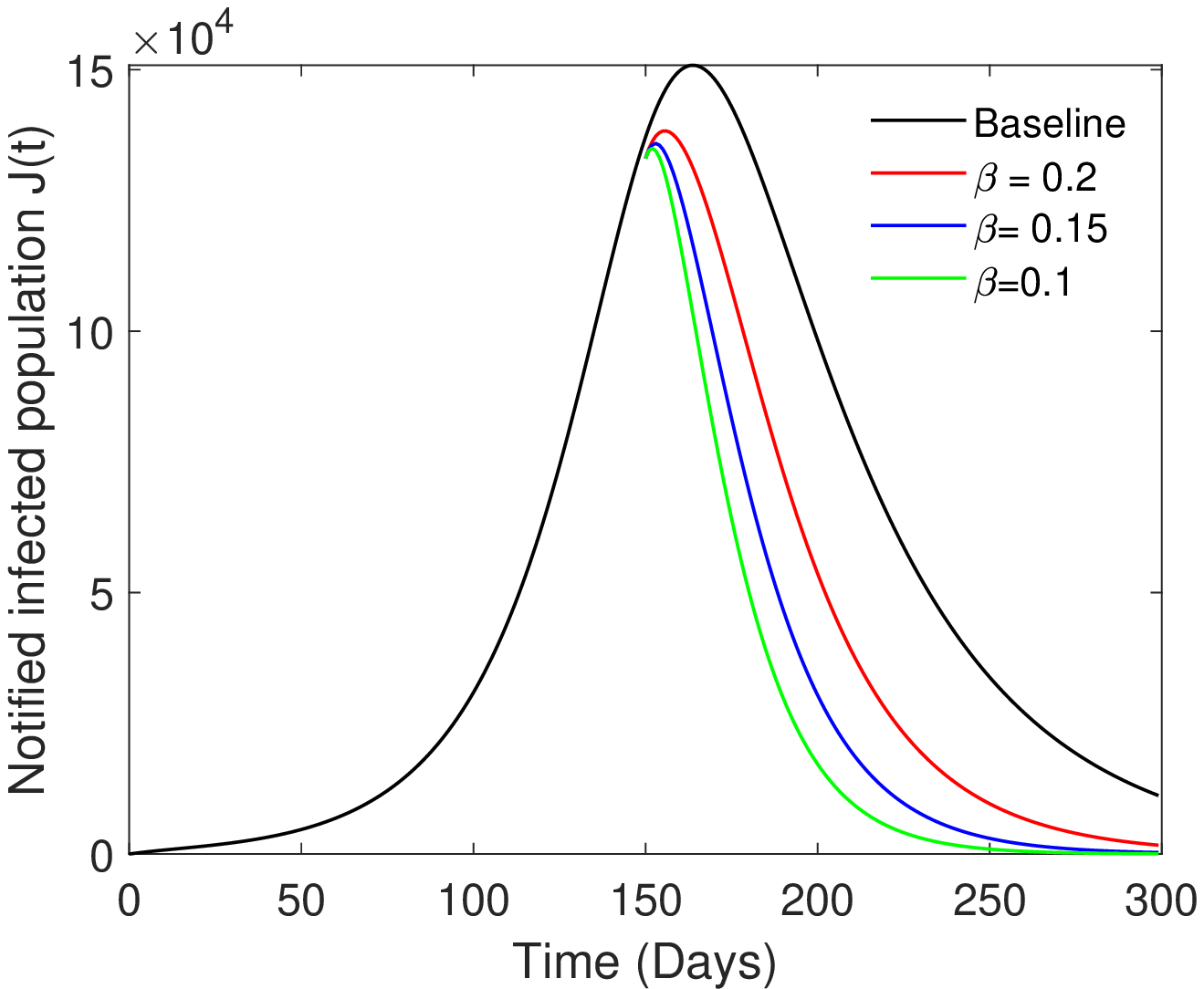}
	\caption{Effect of control parameters $\alpha$ and $\beta$ on the notified and un-notified COVID-19 cases.}
	\label{Fig:prevalance1}
\end{figure}

\begin{figure}[ht]
	\centering
	\includegraphics[width=0.49\textwidth]{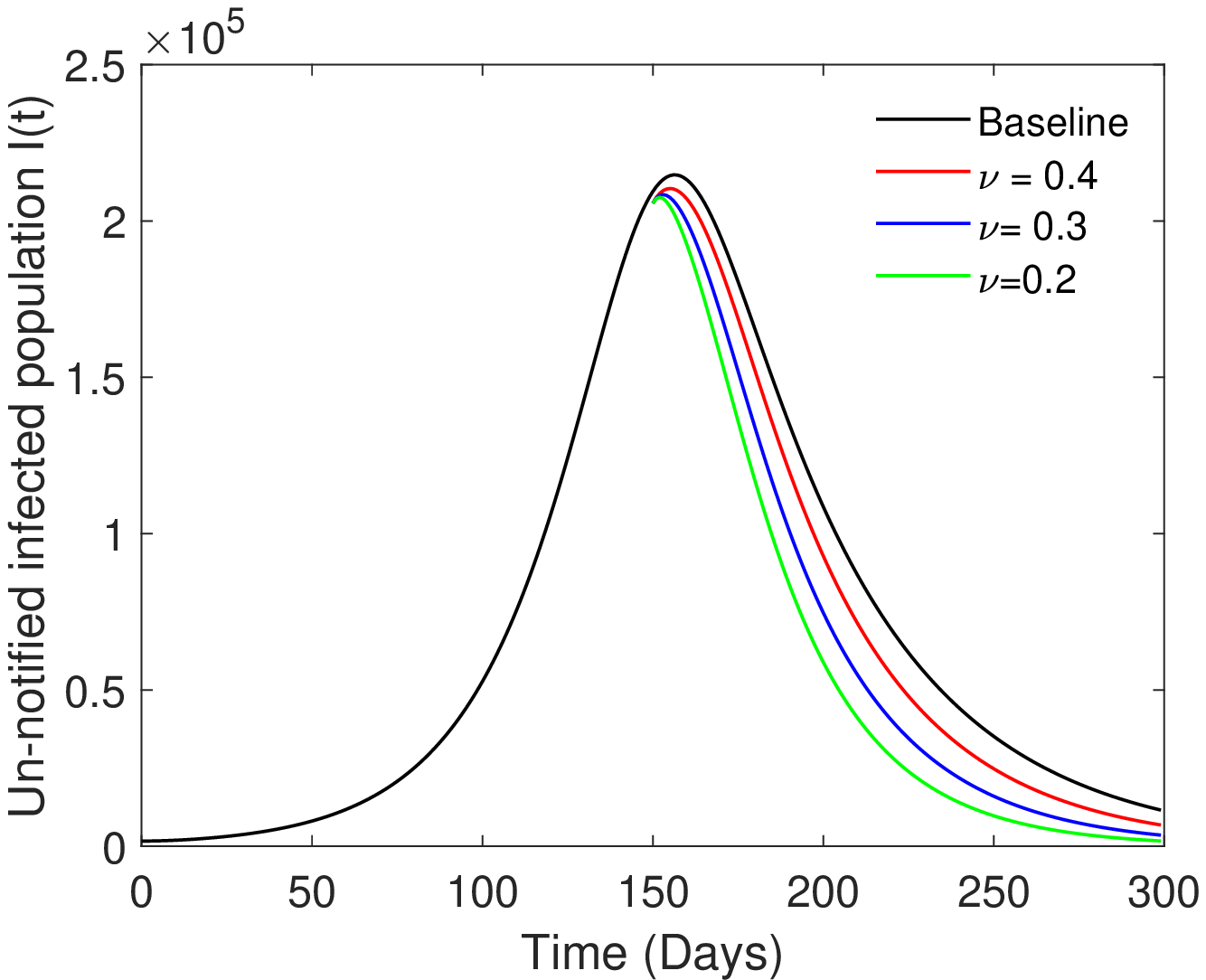}
	\includegraphics[width=0.49\textwidth]{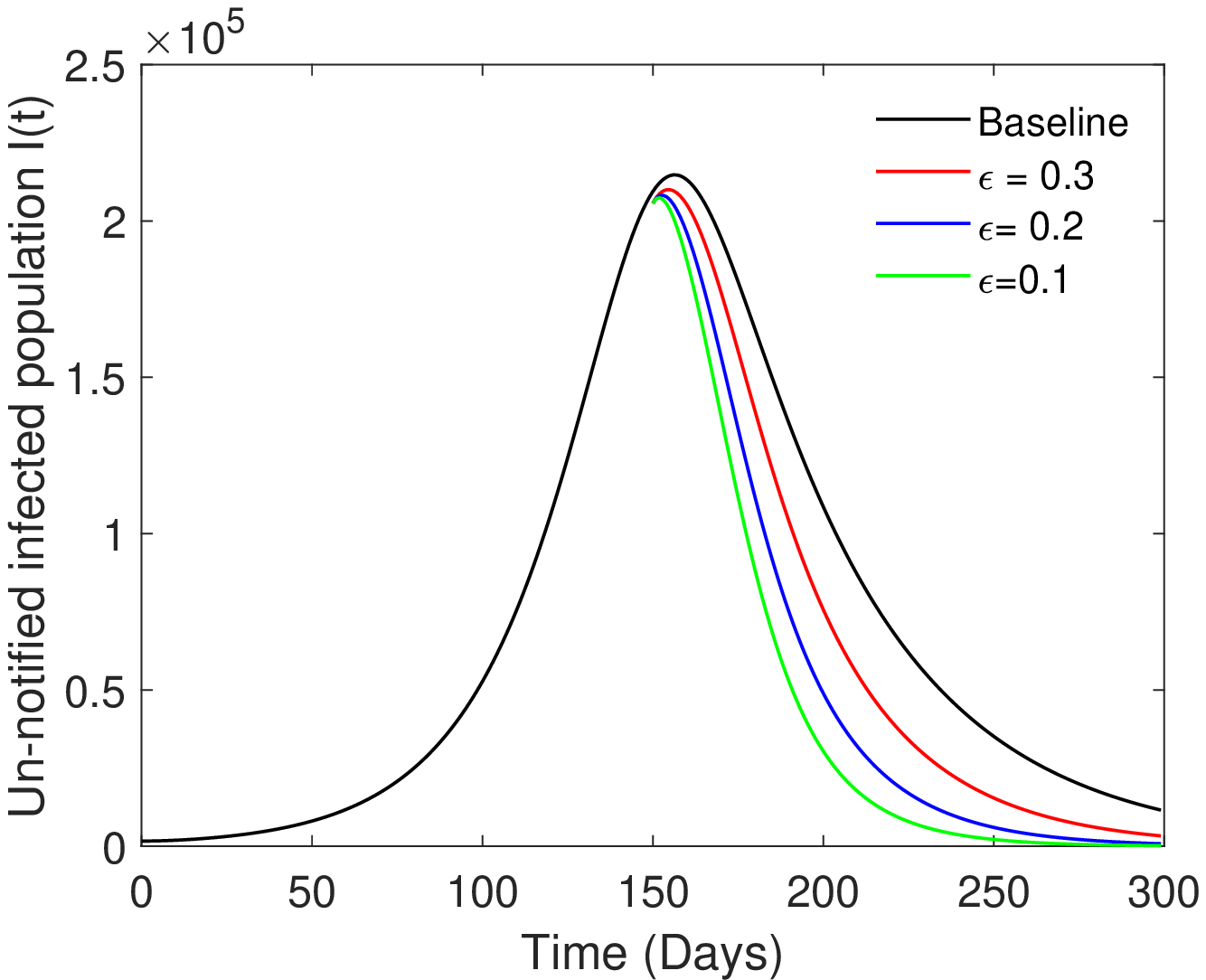}
	\includegraphics[width=0.49\textwidth]{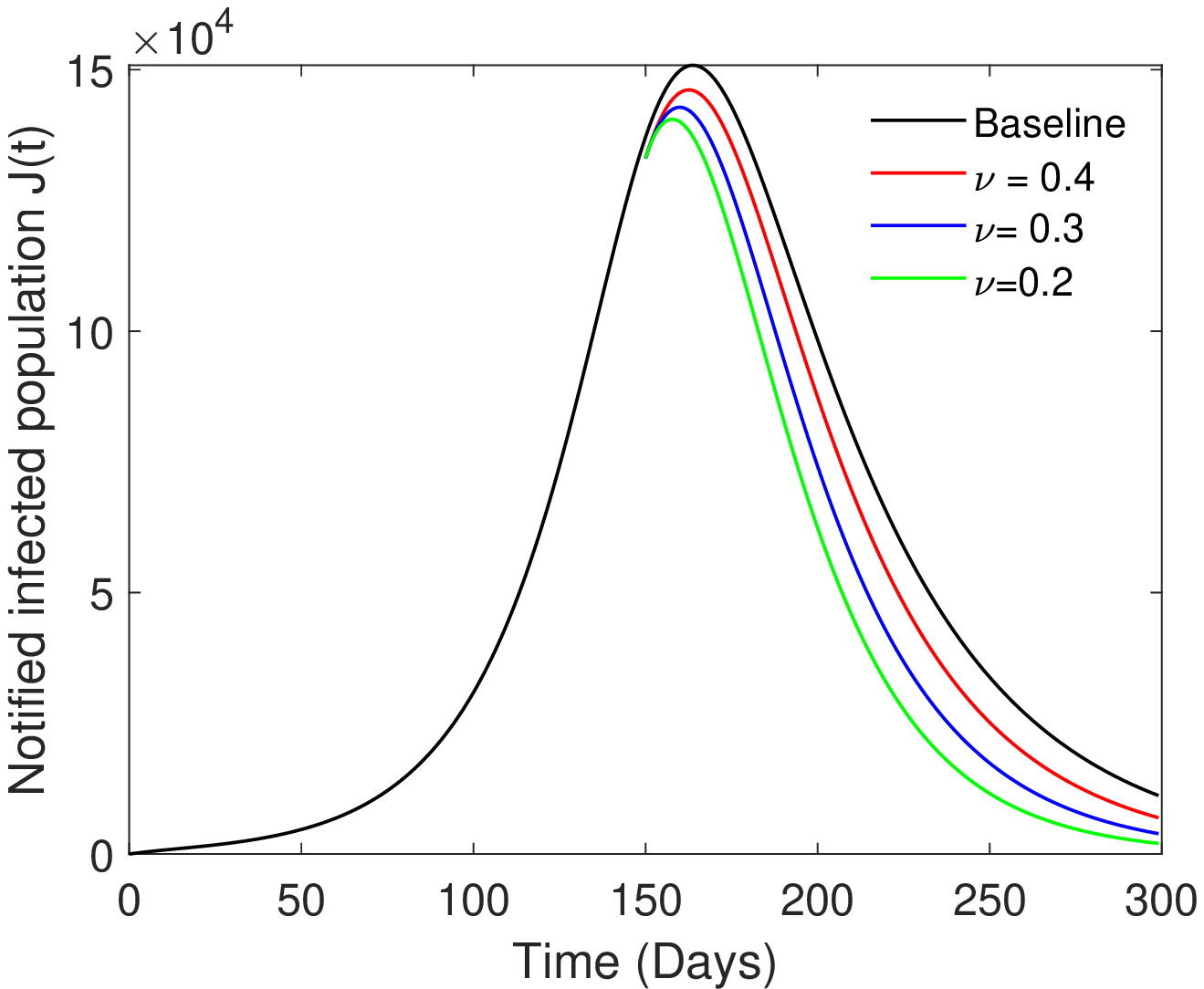}
	\includegraphics[width=0.49\textwidth]{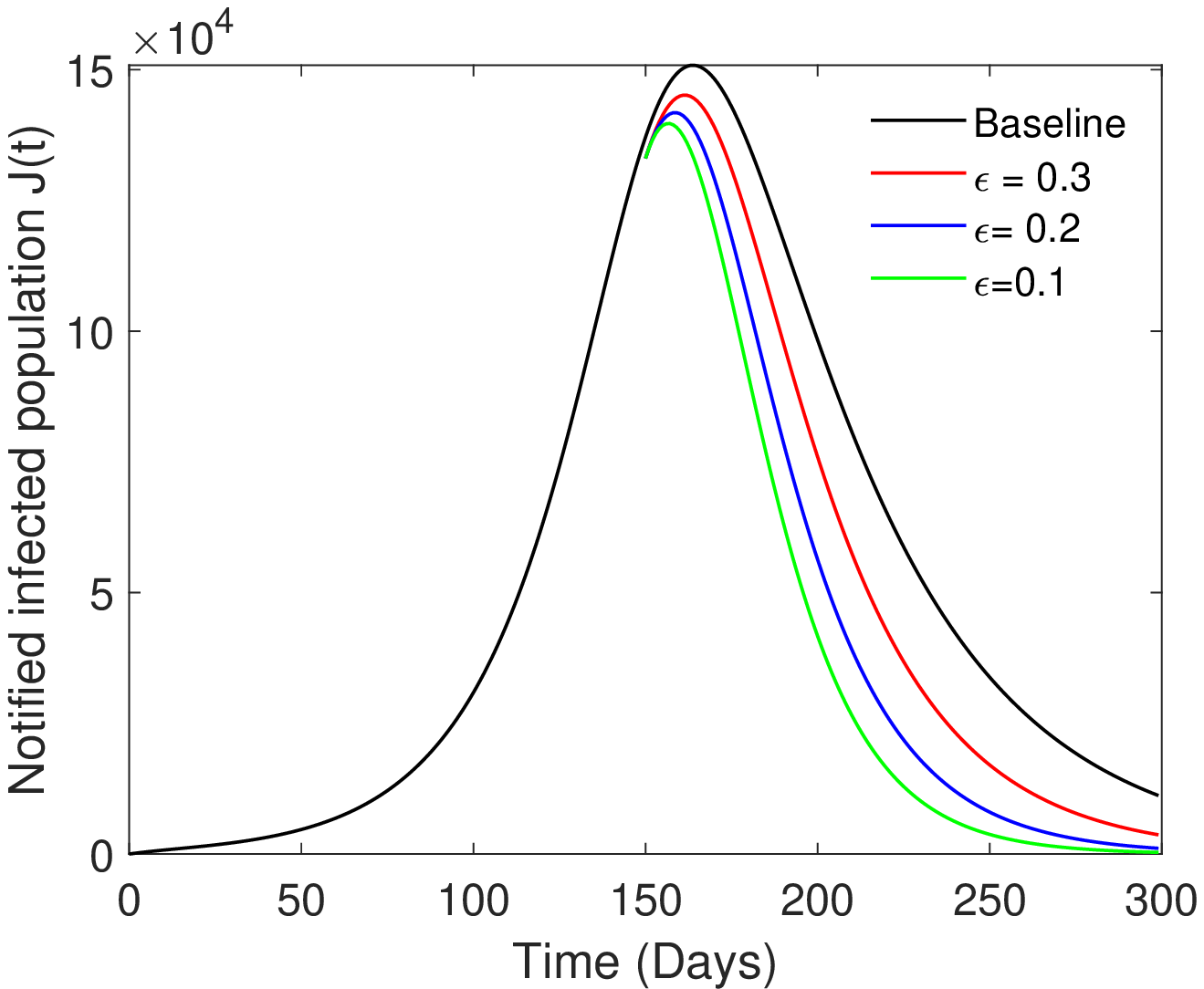}
	\caption{Effect of control parameters $\nu$ and $\epsilon$ on the notified and un-notified COVID-19 cases.}
	\label{Fig:prevalance2}
\end{figure}
It can be observed that all the four controllable parameters show similar trends in case reduction. However, $\beta$ and $\epsilon$ are showing sharp decay in cases.  Note that the scales of case reductions in the Figs. \ref{Fig:prevalance1} and
\ref{Fig:prevalance2} are similar, but we cannot quantify the
effectiveness of the parameters. Therefore, we calculate percentage reductions
in un-notified and notified COVID-19 cases in Colombia, Table
\ref{Tab:percent-reduction-Table}. The fixed parameter values are taken from Table
\ref{tab:mod1} while the values of other parameters are taken from
Table \ref{Tab:estimated-parameters-Table}. We used the following basic formula to compute
percentage reductions in the values of $I(t)$ and $J(t)$

\begin{eqnarray*}
	\textmd{Percentage reduction}=\frac{\textmd{Base value} -
		\textmd{Model output}}{\textmd{Base value}}\times 100.
\end{eqnarray*}

\begin{table}[ht]	
	\centering
	\caption{\bf{Percentage reduction in un-notified and notified COVID-19 cases for different controllable parameter values.}}\vspace{0.3cm}
	\begin{tabular}{cccc} \hline
		\textbf{Parameters} & \textbf{Values} &\textbf{Un-notified case reduction} & \textbf{Notified case reduction} \\ \hline
		$\alpha$ & 0.15 & 18.05 & 16.20 \\ 
		 & 0.3 & 38.73 & 34.71 \\  
		 & 0.6 & 47.91 & 42.98 \\ 
		\hline
        $\beta$ & 0.2 & 42.64 & 38.18 \\ 
        & 0.15 & 61.03 & 54.79 \\  
        & 0.1 & 71.63 & 64.41 \\ 
        \hline
		$\nu$ & 0.4 & 13.46 & 12.08 \\ 
		& 0.3 & 25.75 & 23.00 \\  
		& 0.2 & 35.95 & 32.12 \\ 
		\hline	
		$\epsilon$ & 0.3 & 24.38 & 21.73 \\ 
		& 0.2 & 40.40 & 36.08 \\  
		& 0.1 & 50.99 & 45.65 \\ 
		\hline
	\end{tabular}
	\label{Tab:percent-reduction-Table}
\end{table}

From Table \ref{Tab:percent-reduction-Table}, it can be argued that $\beta$ is most effective in reduction of COVID-19 cases. Maximum reduction in $\epsilon$ (=0.1) can reduce notified COVID-19 cases upto 45.65\%. Learning factor related to aware susceptibles ($\alpha$) also has competitive potential to reduce the COVID-19 cases in the community.

\section{Discussion and conclusion}\label{Discussion}
This paper provides a deterministic model for the transmission dynamics of COVID-19 outbreak incorporating prosocial behaviour. The model, which adopts standard incidence functions in a realistic way, allows COVID-19 to be transmitted by un-notified and notified individuals. To gain insight into its dynamic features, the model was rigorously analyzed. The findings obtained are as follows. There are four type of equilibrium points of the proposed model: awareness free, disease free equilibrium ($E_0$), disease free equilibrium ($E_1$), awareness free endemic equilibrium ($E_2$) and endemic equilibrium point ($E_3$). The awareness free, disease free equilibrium is found to be globally asymptotically stable under a parametric condition (($max\{R_1, \frac{\alpha}{\theta + \mu}\}<1$)). The basic reproduction number ($R_0$) for the proposed model is calculated using the next-generation matrix method. The model has a locally-stable disease-free equilibrium whenever the basic reproduction number is less than unity. The global stability condition for the DFE is also presented using Lyapunov function. The existence and stability criterion of these four equilibria are presented in Table \ref{Tab:stability_existence}. Further, the two threshold quantities $R_1$ and $R_0$ are linearly dependent on each other by the relation $R_0=m_3 R_1$. Using a hypothetical parameter set, we show the stability switch or transcritical bifurcation between $E_0$ and  $E_1$ (see Fig. \ref{Fig:transcritical}(a)). Also, transcritical bifurcations are observed between the equilibria $E_0$ and $E_2$ (see Fig. \ref{Fig:transcritical}(b)) and between $E_1$ and $E_3$ (see Fig. \ref{Fig:transcritical}(c)).

We calibrated the proposed model parameters to fit Colombia's daily cases and death data during the time period of March $19^{th}$, 2020 to August $14^{th}$, 2020. Using test data points (August $15^{th}$, 2020 to August $14^{th}$, 2020), we found that RMSE=2031.5 and MAE=1491.1 for notified case fitting and for fitting COVID-19 deaths, RMSE=137.08 and MAE=125.05 are found. Thus, the fitting is pretty good for actual data in the study period. Using estimated parameters, the basic reproduction number is found to be 0.7815 with 95\% confidence interval (0.7633 -- 0.8014). Which indicate the success of Colombian government in containing this deadly disease. This is also reflected on the current decreasing phase of the notified cases of COVID-19 in Colombia. Global sensitivity analysis is performed with respect to the total number of COVID-19 related deaths. Results indicate that transmission rate ($\beta$), modification factor ($\nu$) and learning factor related to awareness of susceptibles ($\alpha$) are very crucial for reduction in disease related deaths. We have also investigated the impact of four controllable parameters on the prevalence of un-notified and notified COVID-19 cases. From Fig. \ref{Fig:prevalance1} and \ref{Fig:prevalance2}, it can be observed that different level of controls can significantly reduce the burden of COVID-19 from community. However, to better quantify the impacts, we calculate the percentage reduction using a simple formula. This reveals that reduction in transmission rate is most effective in reducing un-notified and notified COVID-19 caese (see Table \ref{Tab:percent-reduction-Table}). Increase in the learning factor ($\alpha$) has competitive potential to flatten the curve. This finding reinforces the need for amplified campaigns and awareness made by individual aware susceptible persons. Prosocial behavior is needed to combat this highly infectious pandemic disease. Along with verified control strategies, the governments should promote prosocial awareness by aware people. This will definitely benefit the case reduction as well as management of future COVID-19 cases.

\bibliographystyle{unsrt}
\biboptions{square}
\bibliography{bib}

\begin{thebibliography}{10}

\bibitem{Who2020}
{WHO. Coronavirus} disease (covid-19) outbreak.
\newblock \url{https://www.who.int/emergencies
  /diseases/novel-coronavirus-2019}, 2020.
\newblock Retrieved : 2020-08-15.

\bibitem{huang2020clinical}
Chaolin Huang, Yeming Wang, Xingwang Li, Lili Ren, Jianping Zhao, Yi~Hu,
  Li~Zhang, Guohui Fan, Jiuyang Xu, Xiaoying Gu, et~al.
\newblock Clinical features of patients infected with 2019 novel coronavirus in
  wuhan, china.
\newblock {\em The Lancet}, 395(10223):497--506, 2020.

\bibitem{gralinski2020return}
Lisa~E Gralinski and Vineet~D Menachery.
\newblock Return of the coronavirus: 2019-ncov.
\newblock {\em Viruses}, 12(2):135, 2020.

\bibitem{cdcgov2020}
Centers for disease control and prevention: 2019 novel coronavirus.
\newblock \url{https: //www.cdc.gov/coronavirus/2019-ncov}, 2020.
\newblock Retrieved : 2020-03-10.

\bibitem{chan2020familial}
Jasper Fuk-Woo Chan, Shuofeng Yuan, Kin-Hang Kok, Kelvin Kai-Wang To, Hin Chu,
  Jin Yang, Fanfan Xing, Jieling Liu, Cyril Chik-Yan Yip, Rosana Wing-Shan
  Poon, et~al.
\newblock A familial cluster of pneumonia associated with the 2019 novel
  coronavirus indicating person-to-person transmission: a study of a family
  cluster.
\newblock {\em The Lancet}, 395(10223):514--523, 2020.

\bibitem{Worldometer2020}
{COVID-19} coronavirus outbreak.
\newblock \url{https://www.worldometers.info/coronavirus/#repro}, 2020.
\newblock Retrieved : 2020-08-15.

\bibitem{gumel2004modelling}
Abba~B Gumel, Shigui Ruan, Troy Day, James Watmough, Fred Brauer, P~Van~den
  Driessche, Dave Gabrielson, Chris Bowman, Murray~E Alexander, Sten Ardal,
  et~al.
\newblock Modelling strategies for controlling sars outbreaks.
\newblock {\em Proceedings of the Royal Society of London. Series B: Biological
  Sciences}, 271(1554):2223--2232, 2004.

\bibitem{li2003angiotensin}
Wenhui Li, Michael~J Moore, Natalya Vasilieva, Jianhua Sui, Swee~Kee Wong,
  Michael~A Berne, Mohan Somasundaran, John~L Sullivan, Katherine Luzuriaga,
  Thomas~C Greenough, et~al.
\newblock Angiotensin-converting enzyme 2 is a functional receptor for the sars
  coronavirus.
\newblock {\em Nature}, 426(6965):450--454, 2003.

\bibitem{de2013commentary}
Raoul~J de~Groot, Susan~C Baker, Ralph~S Baric, Caroline~S Brown, Christian
  Drosten, Luis Enjuanes, Ron~AM Fouchier, Monica Galiano, Alexander~E
  Gorbalenya, Ziad~A Memish, et~al.
\newblock Commentary: Middle east respiratory syndrome coronavirus (mers-cov):
  announcement of the coronavirus study group.
\newblock {\em Journal of virology}, 87(14):7790--7792, 2013.

\bibitem{sardar2020realistic}
Tridip Sardar, Indrajit Ghosh, Xavier Rod{\'o}, and Joydev Chattopadhyay.
\newblock A realistic two-strain model for mers-cov infection uncovers the high
  risk for epidemic propagation.
\newblock {\em PLoS neglected tropical diseases}, 14(2):e0008065, 2020.

\bibitem{cowling2015preliminary}
Benjamin~J Cowling, Minah Park, Vicky~J Fang, Peng Wu, Gabriel~M Leung, and
  Joseph~T Wu.
\newblock Preliminary epidemiologic assessment of mers-cov outbreak in south
  korea, may--june 2015.
\newblock {\em Euro surveillance: bulletin Europeen sur les maladies
  transmissibles= European communicable disease bulletin}, 20(25), 2015.

\bibitem{kim2017middle}
KH~Kim, TE~Tandi, Jae~Wook Choi, JM~Moon, and MS~Kim.
\newblock Middle east respiratory syndrome coronavirus (mers-cov) outbreak in
  south korea, 2015: epidemiology, characteristics and public health
  implications.
\newblock {\em Journal of Hospital Infection}, 95(2):207--213, 2017.

\bibitem{kwok2019epidemic}
Kin~On Kwok, Arthur Tang, Vivian~WI Wei, Woo~Hyun Park, Eng~Kiong Yeoh, and
  Steven Riley.
\newblock Epidemic models of contact tracing: Systematic review of transmission
  studies of severe acute respiratory syndrome and middle east respiratory
  syndrome.
\newblock {\em Computational and structural biotechnology journal}, 2019.

\bibitem{ngonghala2020mathematical}
Calistus~N Ngonghala, Enahoro Iboi, Steffen Eikenberry, Matthew Scotch,
  Chandini~Raina MacIntyre, Matthew~H Bonds, and Abba~B Gumel.
\newblock Mathematical assessment of the impact of non-pharmaceutical
  interventions on curtailing the 2019 novel coronavirus.
\newblock {\em Mathematical Biosciences}, page 108364, 2020.

\bibitem{kucharski2020early}
Adam~J Kucharski, Timothy~W Russell, Charlie Diamond, Yang Liu, John Edmunds,
  Sebastian Funk, Rosalind~M Eggo, Fiona Sun, Mark Jit, James~D Munday, et~al.
\newblock Early dynamics of transmission and control of covid-19: a
  mathematical modelling study.
\newblock {\em The lancet infectious diseases}, 2020.

\bibitem{sardar2020assessment}
Tridip Sardar, Sk~Shahid Nadim, Sourav Rana, and Joydev Chattopadhyay.
\newblock Assessment of lockdown effect in some states and overall india: A
  predictive mathematical study on covid-19 outbreak.
\newblock {\em Chaos, Solitons \& Fractals}, page 110078, 2020.

\bibitem{tang2020estimation}
Biao Tang, Xia Wang, Qian Li, Nicola~Luigi Bragazzi, Sanyi Tang, Yanni Xiao,
  and Jianhong Wu.
\newblock Estimation of the transmission risk of the 2019-ncov and its
  implication for public health interventions.
\newblock {\em Journal of Clinical Medicine}, 9(2):462, 2020.

\bibitem{zhao2020preliminary}
Shi Zhao, Qianyin Lin, Jinjun Ran, Salihu~S Musa, Guangpu Yang, Weiming Wang,
  Yijun Lou, Daozhou Gao, Lin Yang, Daihai He, et~al.
\newblock Preliminary estimation of the basic reproduction number of novel
  coronavirus (2019-ncov) in china, from 2019 to 2020: A data-driven analysis
  in the early phase of the outbreak.
\newblock {\em International Journal of Infectious Diseases}, 92:214--217,
  2020.

\bibitem{asamoah2020global}
Joshua Kiddy~K Asamoah, MA~Owusu, Zhen Jin, FT~Oduro, Afeez Abidemi, and
  Esther~Opoku Gyasi.
\newblock Global stability and cost-effectiveness analysis of covid-19
  considering the impact of the environment: using data from ghana.
\newblock {\em Chaos, Solitons \& Fractals}, page 110103, 2020.

\bibitem{yan2020impact}
Qinling Yan, Yingling Tang, Dingding Yan, Jiaying Wang, Linqian Yang, Xinpei
  Yang, and Sanyi Tang.
\newblock Impact of media reports on the early spread of covid-19 epidemic.
\newblock {\em Journal of Theoretical Biology}, page 110385, 2020.

\bibitem{zhou2020effects}
Weike Zhou, Aili Wang, Fan Xia, Yanni Xiao, and Sanyi Tang.
\newblock Effects of media reporting on mitigating spread of covid-19 in the
  early phase of the outbreak.
\newblock 2020.

\bibitem{chang2020studying}
Xinghua Chang, Maoxing Liu, Zhen Jin, and Jianrong Wang.
\newblock Studying on the impact of media coverage on the spread of covid-19 in
  hubei province, china.
\newblock {\em Mathematical Biosciences and Engineering}, 17(4):3147, 2020.

\bibitem{khajanchi2020dynamics}
Subhas Khajanchi, Kankan Sarkar, Jayanta Mondal, and Matjaz Perc.
\newblock Dynamics of the covid-19 pandemic in india.
\newblock {\em arXiv preprint arXiv:2005.06286}, 2020.

\bibitem{kobe2020modeling}
Fekadu~Tadege Kobe and Purnachandra~Rao Koya.
\newblock Modeling and analysis of effect of awareness programs by media on the
  spread of covid-19 pandemic disease.
\newblock {\em American Journal of Applied Mathematics}, 8(4):223--229, 2020.

\bibitem{mbabazi2020mathematical}
Fulgensia~Kamugisha Mbabazi, Yahaya Gavamukulya, Richard Awichi, Peter
  Olupot-Olupot, Samson Rwahwire, Saphina Biira, and Livingstone~S Luboobi.
\newblock A mathematical model approach for prevention and intervention
  measures of the covid--19 pandemic in uganda.
\newblock 2020.

\bibitem{mohsen2020global}
Ahmed~A Mohsen, Hassan~Fadhil AL-Husseiny, Xueyong Zhou, and Khalid Hattaf.
\newblock Global stability of covid-19 model involving the quarantine strategy
  and media coverage effects.
\newblock 2020.

\bibitem{Wiki2020colombiacovid}
{Wikipedia. Coronavirus} disease (covid-19) outbreak.
\newblock \url{https://en.wikipedia.org/wiki/COVID-19_pandemic_in_Colombia},
  2020.
\newblock Retrieved : 2020-08-20.

\bibitem{pang2020transmission}
LIUYONG PANG, SANHONG LIU, XINAN ZHANG, TIANHAI TIAN, and ZHONG ZHAO.
\newblock Transmission dynamics and control strategies of covid-19 in wuhan,
  china.
\newblock {\em Journal of Biological Systems}, pages 1--18, 2020.

\bibitem{tang2020updated}
Biao Tang, Nicola~Luigi Bragazzi, Qian Li, Sanyi Tang, Yanni Xiao, and Jianhong
  Wu.
\newblock An updated estimation of the risk of transmission of the novel
  coronavirus (2019-ncov).
\newblock {\em Infectious Disease Modelling}, 2020.

\bibitem{frank2020covid}
TD~Frank.
\newblock Covid-19 order parameters and order parameter time constants of italy
  and china: A modeling approach based on synergetics.
\newblock {\em Journal of Biological Systems}, 2020.

\bibitem{rojas2020mathematical}
Jorge~Humberto Rojas, Marlio Paredes, Malay Banerjee, Olcay Akman, and Anuj
  Mubayi.
\newblock Mathematical modeling \& the transmission dynamics of sars-cov-2 in
  cali, colombia: Implications to a 2020 outbreak \& public health
  preparedness.
\newblock {\em medRxiv}, 2020.

\bibitem{bizet2020time}
Nana Geraldine~Cabo Bizet and Dami{\'a}n Kaloni~Mayorga Pe{\~n}a.
\newblock Time-dependent and time-independent sir models applied to the
  covid-19 outbreak in argentina, brazil, colombia, mexico and south africa.
\newblock {\em arXiv preprint arXiv:2006.12479}, 2020.

\bibitem{teheran2020epidemiological}
Anibal~A Teheran, Gabriel Camero, Ronald~Prado de~la Guardia, Carolina
  Hernandez, Giovanny Herrera, Luis~M Pombo, Albert Avila, Carolina Florez,
  Esther~C Barros, Luis~Perez Garcia, et~al.
\newblock Epidemiological characterization of asymptomatic carriers of covid-19
  in colombia.
\newblock {\em medRxiv}, 2020.

\bibitem{manrique2020sir}
Fred~G Manrique-Abril, Carlos~A Agudelo-Calderon, V{\'\i}ctor~M
  Gonz{\'a}lez-Chord{\'a}, Oscar Guti{\'e}rrez-Lesmes, Cristian~F
  T{\'e}llez-Pi{\~n}erez, and Giomar Herrera-Amaya.
\newblock Sir model of the covid-19 pandemic in colombia.
\newblock {\em Revista de Salud P{\'u}blica}, 22(1), 2020.

\bibitem{epstein2008coupled}
Joshua~M Epstein, Jon Parker, Derek Cummings, and Ross~A Hammond.
\newblock Coupled contagion dynamics of fear and disease: mathematical and
  computational explorations.
\newblock {\em PLoS One}, 3(12):e3955, 2008.

\bibitem{just2018oscillations}
Winfried Just, Joan Salda{\~n}a, and Ying Xin.
\newblock Oscillations in epidemic models with spread of awareness.
\newblock {\em Journal of Mathematical Biology}, 76(4):1027--1057, 2018.

\bibitem{martcheva2015introduction}
Maia Martcheva.
\newblock {\em An introduction to mathematical epidemiology}, volume~61.
\newblock Springer, 2015.

\bibitem{may1991infectious}
Robert~M May.
\newblock {\em Infectious diseases of humans: dynamics and control}.
\newblock Oxford University Press, 1991.

\bibitem{samanta2013effect}
Sudip Samanta, Sourav Rana, Anupama Sharma, Arvind~Kumar Misra, and Joydev
  Chattopadhyay.
\newblock Effect of awareness programs by media on the epidemic outbreaks: A
  mathematical model.
\newblock {\em Applied Mathematics and Computation}, 219(12):6965--6977, 2013.

\bibitem{li2020early}
Qun Li, Xuhua Guan, Peng Wu, Xiaoye Wang, Lei Zhou, Yeqing Tong, Ruiqi Ren,
  Kathy~SM Leung, Eric~HY Lau, Jessica~Y Wong, et~al.
\newblock Early transmission dynamics in wuhan, china, of novel
  coronavirus--infected pneumonia.
\newblock {\em New England Journal of Medicine}, 2020.

\bibitem{linton2020incubation}
Natalie~M Linton, Tetsuro Kobayashi, Yichi Yang, Katsuma Hayashi, Andrei~R
  Akhmetzhanov, Sung-mok Jung, Baoyin Yuan, Ryo Kinoshita, and Hiroshi
  Nishiura.
\newblock Incubation period and other epidemiological characteristics of 2019
  novel coronavirus infections with right truncation: a statistical analysis of
  publicly available case data.
\newblock {\em Journal of clinical medicine}, 9(2):538, 2020.

\bibitem{wu2020characteristics}
Zunyou Wu and Jennifer~M McGoogan.
\newblock Characteristics of and important lessons from the coronavirus disease
  2019 (covid-19) outbreak in china: summary of a report of 72 314 cases from
  the chinese center for disease control and prevention.
\newblock {\em Jama}, 323(13):1239--1242, 2020.

\bibitem{yang2020epidemiological}
Penghui Yang, Yibo Ding, Zhe Xu, Rui Pu, Ping Li, Jin Yan, Jiluo Liu, Fanping
  Meng, Lei Huang, Lei Shi, et~al.
\newblock Epidemiological and clinical features of covid-19 patients with and
  without pneumonia in beijing, china.
\newblock {\em Medrxiv}, 2020.

\bibitem{woelfel2020clinical}
Roman Woelfel, Victor~Max Corman, Wolfgang Guggemos, Michael Seilmaier, Sabine
  Zange, Marcel~A Mueller, Daniela Niemeyer, Patrick Vollmar, Camilla Rothe,
  Michael Hoelscher, et~al.
\newblock Clinical presentation and virological assessment of hospitalized
  cases of coronavirus disease 2019 in a travel-associated transmission
  cluster.
\newblock {\em MedRxiv}, 2020.

\bibitem{tindale2020transmission}
Lauren Tindale, Michelle Coombe, Jessica~E Stockdale, Emma Garlock, Wing
  Yin~Venus Lau, Manu Saraswat, Yen-Hsiang~Brian Lee, Louxin Zhang, Dongxuan
  Chen, Jacco Wallinga, et~al.
\newblock Transmission interval estimates suggest pre-symptomatic spread of
  covid-19.
\newblock {\em MedRxiv}, 2020.

\bibitem{lopez2020end}
Leonardo L{\'o}pez and Xavier Rod{\'o}.
\newblock The end of social confinement and covid-19 re-emergence risk.
\newblock {\em Nature Human Behaviour}, 4(7):746--755, 2020.

\bibitem{sun2011effect}
Chengjun Sun, Wei Yang, Julien Arino, and Kamran Khan.
\newblock Effect of media-induced social distancing on disease transmission in
  a two patch setting.
\newblock {\em Mathematical biosciences}, 230(2):87--95, 2011.

\bibitem{castillo2002computation}
Carlos Castillo-Chavez, Zhilan Feng, and Wenzhang Huang.
\newblock On the computation of ro and its role on.
\newblock {\em Mathematical approaches for emerging and reemerging infectious
  diseases: an introduction}, 1:229, 2002.

\bibitem{van2002reproduction}
Pauline Van~den Driessche and James Watmough.
\newblock Reproduction numbers and sub-threshold endemic equilibria for
  compartmental models of disease transmission.
\newblock {\em Mathematical biosciences}, 180(1-2):29--48, 2002.

\bibitem{lasalle1976stability}
Joseph~P LaSalle.
\newblock {\em The stability of dynamical systems}, volume~25.
\newblock Siam, 1976.

\bibitem{smith1995theory}
Hal~L Smith and Paul Waltman.
\newblock {\em The theory of the chemostat: dynamics of microbial competition},
  volume~13.
\newblock Cambridge university press, 1995.

\bibitem{vargas2011global}
Cruz Vargas-De-Le{\'o}n.
\newblock On the global stability of sis, sir and sirs epidemic models with
  standard incidence.
\newblock {\em Chaos, Solitons \& Fractals}, 44(12):1106--1110, 2011.

\bibitem{haario2006dram}
Heikki Haario, Marko Laine, Antonietta Mira, and Eero Saksman.
\newblock Dram: efficient adaptive mcmc.
\newblock {\em Statistics and computing}, 16(4):339--354, 2006.

\bibitem{ghosh2017mathematical}
Indrajit Ghosh, Tridip Sardar, and Joydev Chattopadhyay.
\newblock A mathematical study to control visceral leishmaniasis: an
  application to south sudan.
\newblock {\em Bulletin of mathematical biology}, 79(5):1100--1134, 2017.

\bibitem{marino2008methodology}
Simeone Marino, Ian~B Hogue, Christian~J Ray, and Denise~E Kirschner.
\newblock A methodology for performing global uncertainty and sensitivity
  analysis in systems biology.
\newblock {\em Journal of theoretical biology}, 254(1):178--196, 2008.

\end{thebibliography}

\end{document}